%% file: main.tex
\RequirePackage{fix-cm}
\documentclass[smallcondensed]{svjour3}

\smartqed
\usepackage{graphicx}
\usepackage{lipsum}
\usepackage{ulem}
\usepackage{xcolor}
\usepackage{color}
\usepackage{amssymb}
\usepackage{amsmath}
\usepackage{nicefrac}
\usepackage{algorithmic}
\usepackage{tabularx}
\usepackage{longtable}
\usepackage{empheq}
\usepackage{subfigure}
\usepackage{booktabs}
\usepackage{multirow}
\usepackage[linesnumbered,ruled,resetcount]{algorithm2e}
\usepackage{algorithmic}
\usepackage{setspace}

\newcommand{\reviewAdd}[1]{#1}
\newcommand{\reviewDel}[1]{}

\newcommand{\etal}{$\;$\textit{et\,al.}} 

\newcommand{\Real}{{\rm I\!R}}
\newcommand{\nextnr}{\stepcounter{AlgoLine}\ShowLn}
\DeclareMathOperator*{\argmin}{argmin}
\DeclareMathOperator*{\argmax}{argmax}
\newcommand*\widefbox[1]{\fbox{\hspace{1em}#1\hspace{1em}}}
\newcolumntype{C}[1]{>{\centering\let\newline\\\arraybackslash\hspace{0pt}}m{#1}}
\newcolumntype{R}[1]{>{\raggedleft\let\newline\\\arraybackslash\hspace{0pt}}m{#1}}

\newcount\colveccount
\newcommand*\colvec[1]{
        \global\colveccount#1
        \begin{pmatrix}
        \colvecnext
}
\newcommand{\colvecnext}[1]{
		#1
        \global\advance\colveccount-1
        \ifnum\colveccount>0
                \\
                \expandafter\colvecnext
        \else
                \end{pmatrix}
        \fi
}
\SetKwRepeat{Do}{do}{while}%
\SetKwInput{KwInput}{Input}
\SetKwInput{KwOutput}{Output}
\SetNlSty{color{gray}}{}{}

\AtBeginDocument{%
  \paperwidth=\dimexpr
    1in + \oddsidemargin
    + \textwidth
    + 1in + \oddsidemargin
  \relax
  \paperheight=\dimexpr
    1in + \topmargin
    + \headheight + \headsep
    + \textheight
    + 1in + \topmargin
  \relax
  \usepackage[pass]{geometry}\relax
}

\begin{document}
\normalem
\title{A Dynamic Game Approach for Demand-Side Management\thanks{This work was supported by the Doctoral Training Alliance (DTA) Energy.}
}
\subtitle{Scheduling Energy Storage with Forecasting Errors}

\author{Matthias Pilz         \and
        Luluwah Al-Fagih 
}

\institute{M.Pilz and L.Al-Fagih \at
              Kingston University London, Penrhyn Road, Kingston upon Thames, UK \\
              \email{Matthias.Pilz@kingston.ac.uk}  
}

\date{}

\maketitle

\begin{abstract} 
Smart metering infrastructure allows for two-way communication and power transfer. Based on this promising technology, we propose a demand-side management (DSM) scheme for a residential neighbourhood of prosumers. Its core is a discrete time dynamic game to schedule individually owned home energy storage. The system model includes an advanced battery model, local generation of renewable energy, and forecasting errors for demand and generation. 

We derive a closed-form solution for the best-response problem of a player and construct an iterative algorithm to solve the game. Empirical analysis shows exponential convergence towards the Nash equilibrium. A comparison to a DSM scheme with a static game, reveals the advantages of the dynamic game approach. We provide an extensive analysis on the influence of the forecasting error on the outcome of the game. A key result demonstrates that our approach is robust even in the worst-case scenario. This grants considerable gains for the utility company organising the DSM scheme and its participants.

\end{abstract}

\section{Introduction}
\label{sec:1_introduction}
\input{./sec/1_introduction.tex}

\section{System Model - A Smart Grid Neighbourhood}
\label{sec:2_systemModel}
\input{./sec/2_systemModel.tex}

\section{Dynamic Battery Scheduling Game}
\label{sec:3_game}
\input{./sec/3_DynamicGame.tex}

\section{Results and Discussion}
\label{sec:4_results}
\input{./sec/4_results_newStructure.tex}

\section{Conclusion}
\label{sec:5_conclusion}
\input{./sec/5_conclusion.tex}

\input{./sec/fin.tex}

\bibliographystyle{spmpsci}

\input{./sec/main.bbl}

\renewcommand{\thesection}{\Alph{section}}
\setcounter{section}{1}
\section*{APPENDIX}
\label{sec:appendix}
\input{./sec/appendix.tex}

\end{document}

%% file: sec/1_introduction.tex
Climate change poses a serious threat to the global ecosystem. To limit the increase of global average temperatures, it is critical to restrict greenhouse gas emissions. Currently, burning fossil fuels accounts for the largest share of CO$_2$ emissions by humans into the atmosphere. Renewable energy sources, such as wind and solar, have a much smaller carbon footprint and should be employed instead~\cite{Panwar2011}. Due to the intermittent nature of these sources, their integration into the power system can be a challenging task. Our research investigates possibilities for more efficient and environment friendly access to electricity by means of energy storage and renewable energy generation.

The concept rests upon the implementation of a technologically advanced power grid. In contrast to the current power grid, this smart grid features two-way communication and power transfer between the utility company (UC) and individual households~\cite{Ipakchi2009}. Its decentralised nature is expressed through distributed generation and storage of energy, with individual households capable of doing both.
These households are called prosumers (combination of \textit{pro}ducer and con\textit{sumer}). Moreover, the deployment of smart meters allows households to accurately measure electricity demands in real-time. This permits the implementation of demand-side management (DSM) schemes. Within such schemes, the UC incentivises users to avoid consumption during peak hours by means of dynamic pricing tariffs. These tariffs determine the price per energy unit based on the aggregated load of all users (cf.~\cite{Soliman2014,Ma2017,Celik2017}). This will eventually allow them to reduce investments into fast ramping technologies, needed otherwise.

In~\cite{Soliman2014,Celik2017,Mohsenian-Rad2010,Yaagoubi2015b,Longe2017} consumers react to these price incentives by rescheduling their appliances\reviewAdd{, thus potentially interfering with their habits}.
Among them, \cite{Soliman2014,Celik2017,Longe2017} additionally model the usage of energy storage systems. All of these users are aiming at a reduction of the peak-to-average ratio (PAR) of the aggregated load, since achieving this eventually translates into financial benefits for the participants. 
The methods of choice to obtain the desired schedules are almost always based on game--theoretic concepts. Only~\cite{Longe2017} deviates by using convex optimisation.
Since the DSM scheme directly influences the routines of the users, their comfort levels play an important role. For instance, Yaagoubi \textit{et al.}~\cite{Yaagoubi2015b} found that when acceptable comfort levels are preserved, the amount of savings from the energy bill reduces by more than half of the optimum. Note that all these studies have the common idea of scheduling the usage of appliances and batteries in a day-ahead manner. 

Day-ahead scheduling that does not interfere with the users can solely be realised through energy storage systems. \cite{Nguyen2015,Pilz2017} followed this approach and showed that considerable gains are achievable without interrupting the habits of the consumers. Nguyen \textit{et al.}~\cite{Nguyen2015} put their focus on developing a distributed algorithm, while Pilz \textit{et al.}~\cite{Pilz2017} implemented an advanced battery model, providing insight into how specific battery characteristics influence the participation behaviour and thus the outcome of the game. 

This work builds on these previous results and extends the approach of~\cite{Pilz2017} in two directions. Firstly, we introduce a more sophisticated underlying game structure for the DSM scheme, namely a discrete time dynamic game. \reviewAdd{Within this formulation the action space is continuous instead of the discrete options available in~\cite{Pilz2017}. As a consequence, the outcome for the players improves as they can make more fine-grained decisions. Another advantage is that this allows for the derivation of a best response strategy and thus does not require a computationally expensive search for the best response. }Secondly, we analyse the influence of the forecasting error for demand and energy generation on the scheduling outcome. \reviewAdd{In order to assure the stability of the grid, a}\reviewDel{A} real-world application requires the mechanism to be resilient against eventual errors in the predictions, as they will undoubtedly occur.\\\newpage

Our contributions are as follows:
\vspace{-\baselineskip}
{\setlength{\extrarowheight}{.5em}
	\begin{longtable}{p{.04\textwidth}p{.9\textwidth}}
	(1) & We introduce a novel discrete time dynamic game for energy storage scheduling among prosumers in the smart grid. The closed form solution to the best-response problem is derived by means of a dynamic-programming approach. The ensuing iterative algorithm converges quickly towards the Nash-equilibrium. \reviewDel{A direct comparison to a static game for DSM reveals the superiority of this approach both in terms of computational costs and achieved PAR reduction.}\reviewAdd{Direct comparison to similar approaches, i.e.~\cite{Nguyen2015,Pilz2017,Yaagoubi2015b}, reveals the superiority both in terms of achieved PAR reduction and computational costs.}\\
	(2) & A complete day-ahead DSM scheme, consisting of prosumers with realistically modelled batteries, local renewable energy sources, and forecasting errors for demand and generation is simulated. In contrast to previous works which merely simulate individual days, our scheduling period covers a full year. The length of the simulation allows for an in-depth analysis of the influence of the forecasting errors as well as the impact of the number of participants in the DSM scheme.\\
	(3) & We show that the proposed dynamic game approach is robust with respect to the forecasting errors, even in the worst-case scenario. The respective results exhibit only small deviations in the PAR reduction outcomes compared to runs with accurate predictions, and hardly any influence on the financial benefits for the DSM participants.\\
	(4) & For the first time, a comparison of how different compositions of neighbourhoods perform in the DSM scheme is presented. We find that a community consisting of a mix of consumer types can achieve best results.
	\end{longtable}%
}
\addtocounter{table}{-1}

This paper is organised as follows. In Section~\ref{sec:2_systemModel}, we give an overview of the system, provide details of the DSM protocol, introduce the battery and the renewable energy model, and explain the pricing tariff. Section~\ref{sec:3_game} contains detailed information about the dynamic game. Furthermore, it includes the derivation of the best-response solution and the description of the iterative algorithm. The simulation parameters and the data sets for demand and generation data are presented in the beginning of Section~\ref{sec:4_results}. Then, we compare our approach to the static game approach of~\cite{Pilz2017}, \reviewDel{and }show the influence of the forecasting errors\reviewAdd{, and investigate the neighbourhood composition}. \reviewDel{This section ends with detailed discussions of all the presented results. }Section~\ref{sec:5_conclusion} concludes the paper and points out future research \reviewDel{aspects}\reviewAdd{directions}.

%% file: sec/2_systemModel.tex
In this section, we build the basis to the formulation of the battery scheduling game presented in Section~\ref{sec:3_game}. We introduce the concept of a smart grid neighbourhood that participates in a demand-side management (DSM) program to reduce their electricity bills. Each of the participants is equipped with an individually owned lithium-ion battery in addition to a photovoltaic (PV) cell which generates electricity. Models for both the battery and the PV cell are stated in detail. Moreover, we clarify the specific smart meter infrastructure that is necessary to implement the DSM program, as well as the role of the single utility company (UC) running this program.

\subsection{Neighbourhood and Demand-Side Management Program}
\label{sec:smartGrid}
Consider a residential neighbourhood comprised of $M$ houses. Each of these is equipped with a smart meter. 
\begin{figure}[t]
	\centering
	\includegraphics[width = \linewidth]{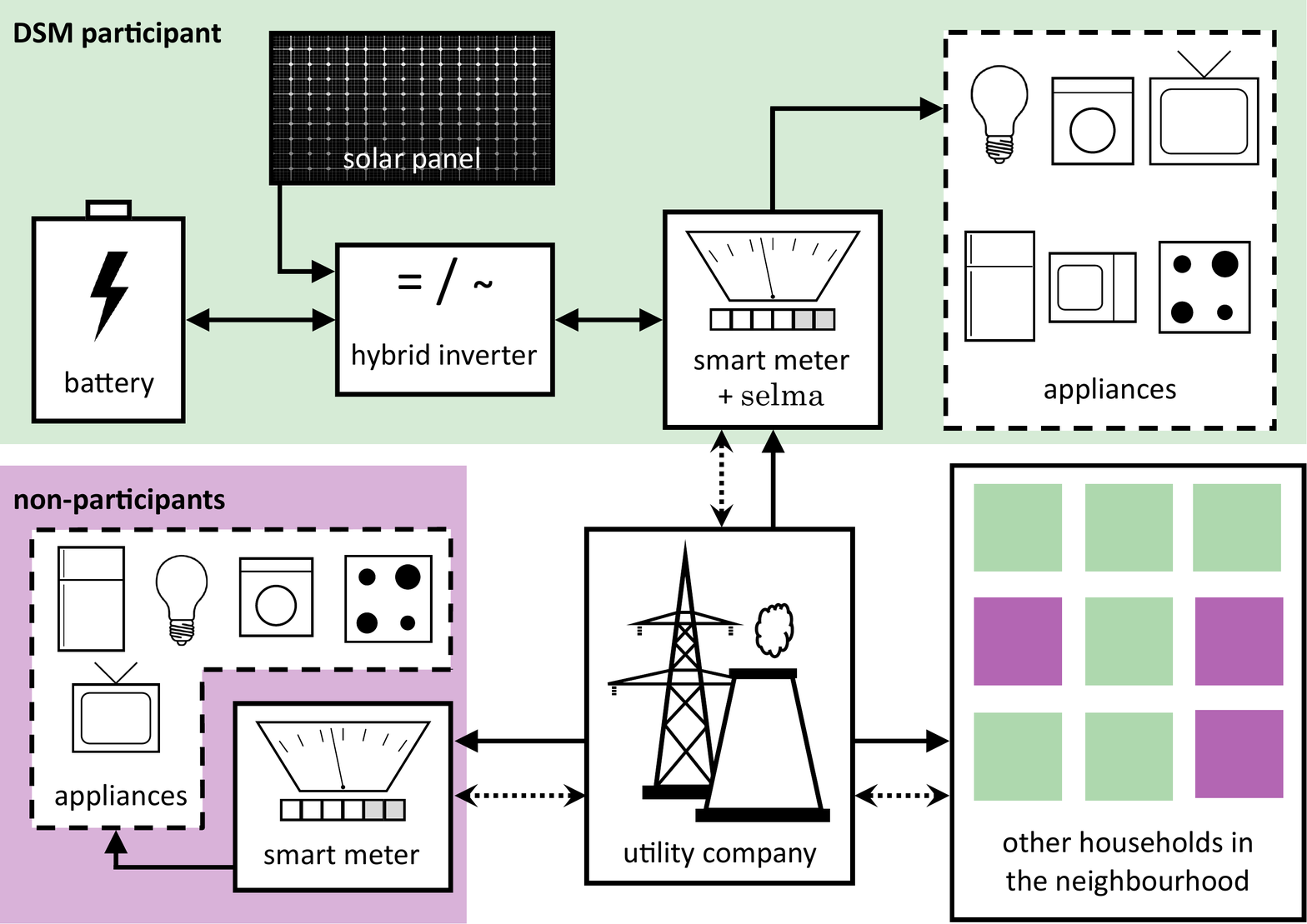}
	\caption{\textit{Systematic sketch of the neighbourhood.} Solid arrows stand for power flow, while dotted arrows stand for information flow.
The top half of the figure represents a participating household of the DSM scheme. It is equipped with a lithium-ion battery and a solar panel. The solar panel can directly charge the battery, but to run any appliance its direct current needs to be converted to alternating current by the inverter. The smart meter collects data and executes the schedule obtained from the \reviewAdd{author's} scheduling software \texttt{selma}\reviewAdd{, which is based on a dynamic game.}
A non-participating household is depicted in the bottom left corner. It is also equipped with a smart meter, collecting data and communicating with the utility company. The complete neighbourhood consists of a number of households (cf.~bottom right) that belong to either one of the shown categories. All of them are served by the same utility company.}
	\label{fig:system}
\end{figure}
Smart meters are capable of measuring electricity consumption accurately and at a higher frequency than the usual monthly or quarterly readings. Furthermore, these devices can communicate directly with the utility company. This eventually allows for the implementation of the DSM program, and also eliminates the need for on-site readings. For our proposed model, we assume that we are able to obtain readings in regular intervals. Based on the reading-frequency, we split each day into $T$ discrete intervals and denote the set of all intervals by $\mathcal{T}$. 

We assume that the $M$ houses are served by the same UC. In order to incentivise consumers to participate in the DSM scheme, the UC offers them a specific pricing scheme, which eventually reduces their electricity bills. Details can be found in Section~\ref{sec:utility}. Let us denote the set of households who participate in the DSM program by $\mathcal{N}\subset\mathcal{M}$, where $\mathcal{M}$ is the set of all households in the neighbourhood. The total number of participants is $N=|\mathcal{N}|$. Besides the different pricing scheme, the participants of the DSM possess their own battery storage system and have solar panels installed. An overview of the neighbourhood is given in Figure~\ref{fig:system}. \\

The DSM scheme can be seen as a protocol, which is gone through repeatedly. In our study the protocol is run once per day. Note that this is a completely automated process run by our scheduling software \texttt{selma} (short for: Scheduler for Electricity in Local MArkets), which needs to be installed on a consumer access device given to each participant of the scheme. \reviewAdd{The algorithm to obtain the schedules is based on a discrete time dynamic game, which will be introduced in Section~\ref{sec:game}.}

Before the start of each scheduling period, \texttt{selma} forecasts the demand\footnote{As of today, the forecasting module is not included into \texttt{selma}. For this study we assume this information to be given and refer the reader to~\cite{Bichpuriya2016,Dolara2015} for details on demand forecasting.} of the respective household for each interval $t\in\mathcal{T}$ of the upcoming day. This information is sent to the UC. The smart meters of non-participants are not able to forecast their own demand. Thus, the UC performs the forecasting step for these households, based on historically collected data. Eventually, forecasted demand curves are aggregated and the information is sent to each DSM participant. Note that no information about individual neighbours is shared, but only aggregated information. This provides anonymity to all consumers.

Based on this input, the households play a dynamic non-cooperative game (cf.~Section~\ref{sec:3_game}). The outcome of the game is a set of schedules, one for each household, which specify how they can make best use of their battery system. The households will follow these schedules throughout the day, even if their actual demand differs from the forecasted one. In Section~\ref{sec:4_results}, we investigate the influence of the forecasting error and show the robustness of the approach even in the worst-case scenario. At the end of the scheduling period, the electricity costs for each consumer is calculated based on the agreed pricing terms and the protocol starts over again.

\subsection{Individual Households}
\label{sec:households}
Households that participate in the DSM scheme are equipped with a lithium-ion battery and PV cells. In this subsection, we introduce the battery model and clarify how the battery can be used. Moreover, details on the PV system are provided. Finally, we clarify the terminology of \textit{demand}, \textit{net-demand} and \textit{load} of a household based on the usage of their battery and PV cells.

\subsubsection{Battery Model and Decision Variables}
\label{sec:battery}
In this paper, we employ the same battery model as used in~\cite{Pilz2017}. This includes charging, discharging, and self-discharging characteristics of a lithium-ion battery. In fact, the same model may also be applied for lead--acid battery systems (but not nickel--based batteries due to their different charging behaviour). As all our simulations are based on a real-world lithium-ion battery system, in the following we will only refer to them as such. 

\paragraph{Charging:}
Lithium-ion batteries are charged in a two--stage process~\cite{Richtek2014}. In the first stage, the state-of-charge (SOC) increases linearly. This stage is called the `constant current' (CC) stage, with a charging rate limited by $\rho^+>0$. In the second stage, i.e.~the `constant voltage' (CV) stage, the effective charging rate levels off exponentially towards the point where the SOC reaches the nominal maximum capacity $s_{\max}$ of the battery. The point of transition from the first stage to the second is indicated by a SOC $s^*$ and an associated time $t^*$, which needs to be specified for the respective battery. During both stages, we additionally consider losses due to the specific charging efficiency $\eta^+$ with $0\leq\eta^+\leq 1$. Additionally, certain losses occur from the hybrid inverter (cf.~Figure~\ref{fig:system}), modelled by $\eta_{\text{inv}}$ with $0\leq\eta_{\text{inv}}\leq 1$. The hybrid inverter transforms the direct current from either the battery or PV into alternating current at usable voltage and frequency for the household appliances. It also works in the reverse direction to charge the battery.

To obtain an insight into how the households can make use of their battery system, let us look at a specific example (cf.~Figure~\ref{fig:charging}). Given a certain value for the SOC, e.g.~$s'$, we can associate a time $t'$, and thus specify a point on the charging curve.
\begin{figure}[t]
	\centering 
	\subfigure[charging\label{fig:charging}]{\includegraphics[width=0.495\linewidth]{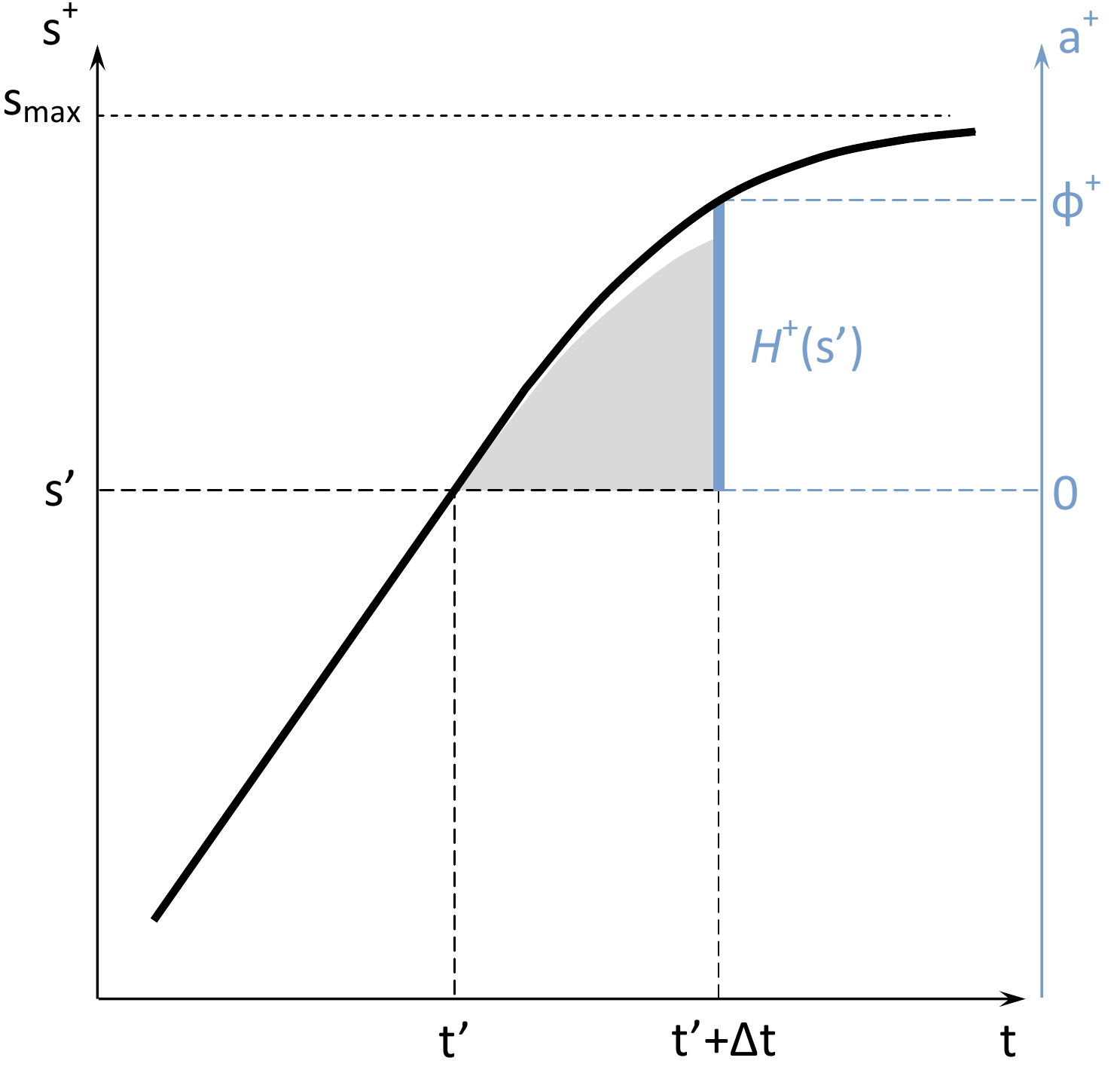}}
	\subfigure[discharging\label{fig:discharging}]{\includegraphics[width=0.495\linewidth]{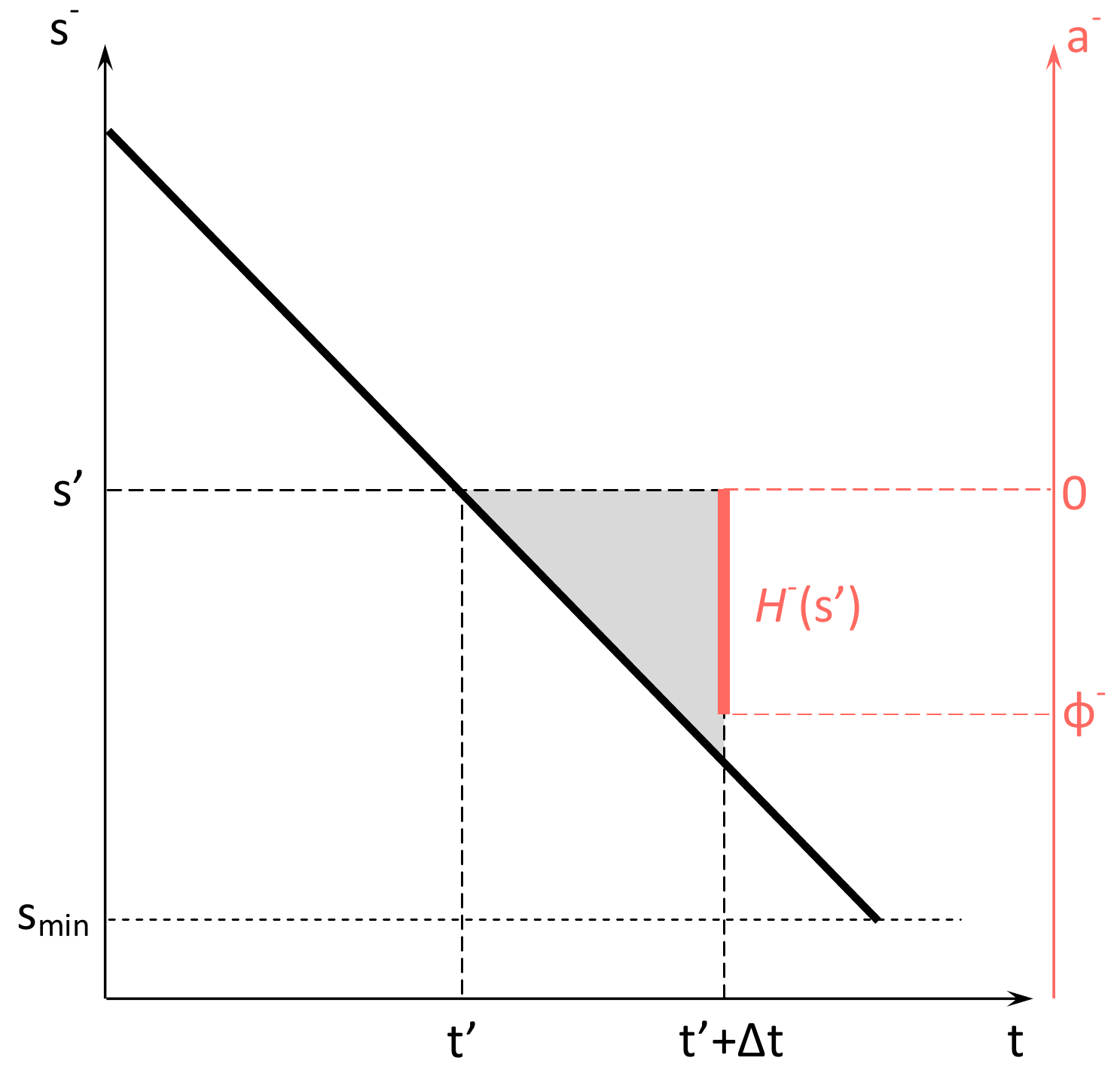}}
	\vspace{-0.5\baselineskip}
\caption{\textit{Schematic illustration of the charging and discharging behaviour of a lithium-ion battery.} The graph on the left (right) shows the characteristic charging (discharging) curve. Given a certain state of charge $s'$, the grey area stands for the achievable state of charge within the following interval. The right axes represent the possible decisions when charging (discharging), where $H^+\left(s'\right)$  ($H^-\left(s'\right)$) summarise the decision intervals. The discrepancy between the achievable state of charge and the decision interval are due to losses, i.e.~imperfect efficiencies while charging (discharging).}
\label{fig:chargingAndDischarging}
\end{figure}

Within the next interval of length $\Delta t$, the decision variable $a^+$ of how much to charge the battery will lie in $\mathcal{H}^+\left(s'\right)=\left\{a^+| h^+\left(s',a^+\right)\leq 0 \right\}$, with
\begin{equation}
	 h^+\left(s',a^+\right) = \colvec{2}{-a^+}{a^+ - \phi^+\left(s'\right)}\ .
	 \label{eqn:charging_restriction}
\end{equation}
In other words, $a^+$ is limited by $0<a^+\leq \phi^+\left(s'\right)<s_{\max}-s'$. We use the notation above to comply with the one shown in~\cite{Nie2006}.
The upper limit $\phi^+\left(s'\right)$ is described by the charging curve, as described above, 
\begin{equation}
	\phi^+\left(s'\right) = \begin{cases}
		\rho^+\Delta t & \text{if CC charged} \\
		s_{\max}\gamma_1\exp\left[-\frac{\Delta t}{\gamma_2}\right] & \text{if CV charged}		
	\end{cases}\ ,
\end{equation}
where $\gamma_1, \gamma_2$ are defined such that the charging curve is smooth at the transition point $(t^*,s^*)$. The discrepancy between the grey--shaded area and the charging curve in Figure~\ref{fig:charging} results from an imperfect charging efficiency. In fact, based on the decision variable $a^+$ the SOC of the battery changes according to the charging transition equation
\begin{equation}
	s\left(t'+\Delta t\right) = s\left(t'\right) + \eta_{\text{inv}}\;\eta^+ a^+\ .
	\label{eqn:chargingTransition}
\end{equation}

\paragraph{Discharging and Self-Discharging:}
We model the discharging behaviour of lithium-ion batteries by a linear decrease in the SOC. Here, the slope is given by the discharging rate $\rho^-<0$. In order to account for the usual sharp drop off of the discharging rate at low capacities, discharging is prohibited below a minimum SOC $s_{\min}$. Again, we also consider losses due to the specific discharging efficiency $\eta^-$ with $0\leq\eta^-\leq 1$ and the hybrid inverter.

In Figure~\ref{fig:discharging} a specific example is given, to clarify how the user can discharge its battery. Within the respective interval, the decision variable $a^-$ of how much to discharge the battery will lie in $\mathcal{H}^-\left(s'\right)=\left\{a^-| h^-\left(s',a^-\right)\leq 0 \right\}$, with
\begin{equation}
	 h^-\left(s',a^-\right) = \colvec{2}{a^-}{-a^- + \phi^-\left(s'\right)}\ .
	 \label{eqn:discharging_restriction}
\end{equation}
In other words, $a^-$ is limited by $s'-s_{\min}<\phi^-\left(s'\right)\leq a^- < 0$ and 
\begin{equation}
	\phi^-\left(s'\right) = \rho^-\Delta t\;\eta_{\text{inv}}\;\eta^-\ .
	\label{eqn:phiMinus}
\end{equation} 
The dependency on $s'$ in \eqref{eqn:phiMinus} is implicitly given by the fact that we cannot go lower than $s_{\min}$. Note that $\phi^-$ also depends on the efficiency parameter, such that the actual amount taken from the battery in correspondence with the decision variable $a^-$ (grey--shaded area in Figure~\ref{fig:discharging}) is given by the discharging transition equation
\begin{equation}
	s\left(t'+\Delta t\right) = s\left(t'\right) + \frac{a^-}{\eta_{\text{inv}}\;\eta^-}\ .
	\label{eqn:dischargingTransition}
\end{equation}
In the following subsection, we will see that $\phi^-$ is additionally limited by the demand of the specific household, i.e.~one can only discharge as much as is needed to run all appliances.

Whenever the battery is neither charging nor discharging, it will be subject to self-discharging. We model this type of behaviour with an exponential decline. This case corresponds to the decision variable $a = 0$. The respective self-discharging transition equation is given by
\begin{equation}
	s\left(t'+\Delta t\right) = s\left(t'\right)\cdot\left(1 + \bar{\rho}\right)^{\Delta t}
	\label{eqn:selfDischargingTransition}
\end{equation}
where $\bar{\rho}<0$ is the self-discharging rate.\\

For later usage (cf.~Section~\ref{sec:game}), we summarise the transition equations for charging, discharging and self-discharging into a single transition equation $f$, i.e.
\begin{equation}
	s(t+\Delta t) = f\left(s(t),a\right) = \begin{cases}
		s(t) + \eta_{\text{inv}}\;\eta^+ a &\ , a>0 \\
		s(t) + \nicefrac{a}{(\eta_{\text{inv}}\;\eta^-)} &\ , a<0 \\
		s(t)\cdot\left(1 + \bar{\rho}\right)^{\Delta t} &\ , a=0
	\end{cases}\ .
	\label{eqn:transSum}
\end{equation}
Furthermore, we combine the restrictions of the decision variable due to the battery restrictions for charging and discharging, i.e.
\begin{equation}
	h(s,a) = \colvec{2}{a - \phi^+\left(s\right)}{-a + \phi^-\left(s\right)}\ .
	\label{eqn:restrictionComb}
\end{equation}

\subsubsection{PV Model}
\label{sec:pv}
We model the solar panel as an additional source of electricity besides the grid connection. The output of the $n$th household's PV system during interval $t$ is denoted by $w^t_n$. It can serve two purposes: (i) direct usage by household appliances, and (ii) charging the battery. Whereas direct usage is influenced by the efficiency of the hybrid inverter, charging the battery does not require any inversion and thus only depends on the charging efficiency of the battery. 

An important parameter of the PV installation is the nominal kilowatt peak $kW\!p$ of the system. It is a measure of the size of the system and denotes the maximum output that can be expected under standardised conditions. A PV system which operates at its maximum capacity, e.g.~$kW\!p =3\;$kW, for one hour will produce $3\;$kWh. Note that identifying the optimal size of the PV installation does not fall within the scope of this article. An approximated scale is obtained from~\cite{Zhang2016d,Olaszi2017} (cf.~Section~\ref{sec:simSetup}).

\subsubsection{Demand, Net-Demand and Load}
\label{sec:demandLoad}
We define the demand $\bar{d}^t_m \geq 0$ of a household $m\in\mathcal{M}$ as the amount of electricity that is needed to run all its appliances during the time interval $t\in\mathcal{T}$. Thus, the total daily demand-schedule can be written as $\bar{d}_m=\left(\bar{d}^0_m,\dots,\bar{d}^{T-1}_m\right)$. Throughout the paper, we assume that the demand cannot be shifted. Thus our approach is fully non-intrusive and does not influence the behaviour of the user. 

Combining the demand $\bar{d}^t_n$ of a household $n\in\mathcal{N}$ with the generated electricity $w_n^t$ from the solar panel, gives the net-demand 
\begin{equation}
	d^t_n = \bar{d}^t_n-\eta_{\text{inv}}\;w^t_n\ , 
	\label{eqn:netDemand}
\end{equation}
where $\eta_{\text{inv}}$ is the efficiency of the inverter (cf.~Figure~\ref{fig:system}). Theoretically, this value can be smaller than zero, i.e.~when the effective generation is larger than the demand in the specific interval. Practically, we ensure $d^t_n \geq 0$ by storing all excess energy directly in the battery. For households $m\not\in\mathcal{N}$, that do not participate in the DSM scheme, the net-demand is identical to the demand.

Let $l_m^t$ denote the load, i.e.~the amount of energy drawn from the grid by household $m\in\mathcal{M}$ during interval $t\in\mathcal{T}$. For households which do not participate in the DSM scheme, the load equals their demand. For the others, the load depends on the decision $a_n^t$ taken at the specific interval. In other words, it combines the net-energy demand with the amount of energy that is charged or discharged by the battery
\begin{equation}
	\l_n^t = d^t_n + a^{t}_n\ ,
	\label{eqn:load}
\end{equation}
where $\max\left\{-d^t_n, \phi^-\right\}\leq a^{t}_n \leq \phi^+$. The lower boundary expresses the fact that one cannot discharge more than is actually needed to fulfil the net-demand, while at the same time all battery restrictions remain valid. Due to this condition and \eqref{eqn:netDemand}, we ensure that $l_m^t\geq 0$ for all $m\in\mathcal{M}$ and all intervals $t\in\mathcal{T}$. We write $l_m=\left(l^0_m,\dots,l^{T-1}_m\right)$ for the schedule of loads of a specific household. Furthermore, we can calculate the total load on the grid for interval $t$ by
\begin{equation}
	L^t = \sum_{m\in\mathcal{M}} l_m^t\ .
	\label{eqn:totalLoad}
\end{equation}
Similarly, we define the average aggregated load of all households other than $n$ during time interval $t$ by $L^t_{-n} = \nicefrac{1}{(M-1)}\sum_{m\in\mathcal{M}\setminus n} l_m^t\ .$

\subsubsection{Forecasting Errors}
\label{sec:error}
The DSM protocol states, that households send a forecast of their net-demand to the UC. This depends on the demand as well as the electricity generated by the solar panel. Both variables will introduce errors, that need to be accounted for. In this paper, we consider the worst-case scenario. \cite{Bichpuriya2016} gives a comprehensive overview of current techniques for short term demand forecasting. \reviewAdd{They specifically investigate how combining forecasts obtained from an integrated auto-regressive moving average, an artificial neural network, and a similar day approach can improve the short term load forecast.}
From~\cite{Bichpuriya2016}, we obtain an upper limit for the forecasting error $\epsilon_d$, expressed as a percentage of the actual demand.
Similarly, \cite{Dolara2015} gives an insight into 24 hour PV power output prediction. The forecasting error $\epsilon_w$ is also given as a percentage of the actual generation.
 
The worst-case scenario is constituted when these two errors carry opposing signs \reviewAdd{and are correlated between all the participants}. This becomes clear from~\eqref{eqn:netDemand}, since both contributions for the net-demand enter with different signs. Intuitively, it makes sense that in the worst-case the forecasted net-demand is smaller than the actual demand. This is because a too small forecasted net-demand does disguise the incentive to make use of the battery system. With the same argument, the worst-case solar forecast is higher than the actual one. It might imply a sufficient SOC of the battery, when in reality more charging would have been necessary.

\subsection{The Utility Company}
\label{sec:utility}
Throughout the paper, we assume a single utility company (UC) serves all the consumers in the neighbourhood. The UC runs a DSM scheme in order to reshape the load profile. To be more precise, they want to achieve a flatter profile such that investments into fast ramping technology, which is needed to deliver peak demand, can be reduced. The incentive for the users to limit consumption during peak hours is given by a dynamic pricing tariff: The cost per energy unit is calculated separately for each interval and depends on the aggregated load of all users in the neighbourhood. Following~\cite{Pilz2017,Mohsenian-Rad2010,Yaagoubi2015b,Nguyen2015}, we employ a quadratic cost function $g^t$:
\begin{equation}
	g^t(y) = c_2\cdot y^2 + c_1\cdot y + c_0\ ,\ \ t\in\mathcal{T}\ ,
	\label{eqn:costFunction}
\end{equation}
where $y$ is the aggregated load at time $t$ given by $L^t$ and the coefficients $c_2>0$, $c_1\geq 0$ and $c_0\geq 0$. Similar to~\cite{Pilz2017,Soliman2014,Mohsenian-Rad2010}, we employ a proportional billing scheme, where each participant of the DMS scheme pays for their share of the consumption, i.e. the electricity bill $B_n$ yields
\begin{equation}
	B_n = -\Omega_n \sum_{t\in\mathcal{T}} g^t\ \ \forall n\in\mathcal{N}\ ,
	\label{eqn:billDSM}
\end{equation}
with
\begin{equation}
	\Omega_n = \frac{\sum_{t}l_n^t}{\sum_{t}\sum_{k}l_k^t}\ .
	\label{eqn:proportionalFactor}
\end{equation}

For households that do not participate in the DSM scheme, a standard fixed-price tariff is employed, i.e.
\begin{equation}
	B_m = p \sum_{t}l^t_m\ \ \forall m\in\mathcal{M}\setminus \mathcal{N}\ .
	\label{eqn:billNoDSM}
\end{equation}

%% file: sec/3_DynamicGame.tex
In this section, we formulate the non-cooperative dynamic game between the households that possess individual energy storage and photovoltaic (PV) installations. To do so, we introduce the relevant notation and relate it to their respective `real-world' meaning according to our system (cf.~Section~\ref{sec:2_systemModel}). Furthermore, the notion of a Nash equilibrium (NE) is defined and an important result concerning the link between the NE for the whole game and the NE for a subgame is provided. Subsequently a dynamic programming algorithm is presented from which we derive a closed form expression of the best response, i.e.~the best decision a player can make in response to fixed decisions of other players. Eventually we use this result to construct an iterative algorithm that computes \reviewDel{the }\reviewAdd{a }NE of the game. 

\subsection{Definitions and Game Formulation}
\label{sec:game}
Formally, the game belongs to the category of discrete time dynamic games (cf.~\cite{Nie2006}), where players make their decisions sequentially in stages. These stages directly correspond to the daily intervals introduced in Section~\ref{sec:smartGrid}. For each stage we define a state of the game, i.e.~the current state-of-charge (SOC) of all batteries, representing the configuration of the overall system. Furthermore, we define a transition equation that models the evolution of this state based on the decisions of the players. In other words, the players will choose actions that are directly related to their battery usage, which in turn depends on the state of the game. We consider a game with open-loop information structure, which means that the initial state of the game is known by all players. In this game, players want to minimise their energy bill, i.e.~their utility function, which depends not only on their own but also on the decisions of all other players. In a nutshell, we have:
\begin{definition}
	Our discrete time dynamic game with open-loop information structure consists of the following components:\\
{\setlength{\extrarowheight}{.5em}
	\begin{longtable}{p{.04\textwidth}p{.9\textwidth}}
	(1) & A set of \textit{players}, i.e.~participating households (cf.~Section~\ref{sec:smartGrid}), $\mathcal{N} = \{1,2,\dots,n,\dots,N\}$, where \reviewAdd{$N$} denotes the number of players.\\
	(2) & A set of \textit{stages}, i.e.~intervals (cf.~Section~\ref{sec:smartGrid}), $\mathcal{T} = \{0,1,\dots,t,\dots,T-1\}$, where \reviewAdd{$T$} denotes the number of stages and thus the number of decisions a player can make in the game. \\
	(3) & Scalar \textit{state variables} $s_n^t\in\mathcal{S}_n\subset\Real$ denoting the SOC of the $n$th player's battery at stage $t\in\mathcal{T}\cup\{T\}$. Collectively, we denote the state variables of all players at stage $t$ by $s^t:=\left(s_1^t,s_2^t,\dots,s_N^t\right)\in\mathcal{S}:=\mathcal{S}_1\times\mathcal{S}_2\times\cdots\times\mathcal{S}_N\subset\Real^N$. In the open-loop information structure it is assumed that the \textit{initial state} $s^0$ is known\footnote{\reviewAdd{Later we will see that the solutions/schedules require the players to deplete their battery towards the end of the scheduling period (cf.~finite horizon effect) to achieve maximum utility. This means as long as none of the players deviates from their respective schedule this knowledge is implicitly shared.}} to all players $n\in\mathcal{N}$.\\
	(4) & Scalar \textit{decision variables} $a_n^t\in\mathcal{H}_n^t\left(s_n^t\right)\subset\mathcal{A}_n\subset\Real$ (for definition of $\mathcal{H}_n^t$ see item (5)) denoting the usage of the battery of the $n$th player at time $t\in\mathcal{T}$. Collectively, we denote the decision variables of all players at stage $t$ by $a^t:=\left(a_1^t,a_2^t,\dots,a_N^t\right)\in\mathcal{A}:=\mathcal{A}_1\times\mathcal{A}_2\times\cdots\times\mathcal{A}_N\subset\Real^N.$ Furthermore we define the \textit{schedule of battery usage} of an individual player $n\in\mathcal{N}$ as a collection of all its decisions in the stages of the game by $a_n:=\left(a_n^0, a_n^1,\dots,a_n^{T-1}\right)$. A \textit{strategy profile} is denoted by $a:=\left( a_1,a_2,\dots,a_N\right)$.\\
	(5) & A set of \textit{admissible decisions} $\mathcal{H}_n\left(s_n^0\right) := \left\{a_n~|~h_n^t\left(s_n^t,a_n^t\right)\leq 0,\ t\in\mathcal{T}\right\}\subset\Real^T$ for the $n$th player. The function $h_n^t\left(s_n^t,a_n^t\right)$ has been defined in \eqref{eqn:restrictionComb} Section~\ref{sec:battery}, capturing the restrictions posed on the battery. We denote $\mathcal{H}_n^t\left(s_n^t\right):= \left\{a_n^t~|~h_n^t\left(s_n^t,a_n^t\right)\leq 0\right\}\subset\Real$\\
	(6) & A \textit{state transition equation}
		\begin{equation}
			s_n^{t+1} = f_n^t\left(s_n^t, a_n^t\right),\ \ t\in\mathcal{T},\ n\in\mathcal{N},
		\label{eqn:stateTransistion}
		\end{equation}
		governing the state variables $\left\{s^t \right\}_{t=0}^T$. The function $f_n^t\left(s_n^t, a_n^t\right)$ is the discretised version of \reviewDel{the }the transition equation \eqref{eqn:transSum} defined in Section~\ref{sec:battery}, showing how a decision of the player influences the state of its battery for the upcoming stage.\\
	(7) & A \textit{stage additive utility function} 
		\begin{equation}
		U_n\left(s_n^0, \left(a_n, a_{-n}\right)\right)=\reviewAdd{-}g_n^T\left(s_n^T\right) \reviewDel{+ }\reviewAdd{- }\sum_{t=0}^{T-1}g_n^t\left(s_n^t, \left(a_n^t, a_{-n}^t\right)\right)
		\end{equation}				
		 for the $n$th player, where $a_{-n}:=\left(a_1,a_2,\dots,a_{n-1},a_{n+1},\dots,a_N\right)$ denotes the decisions of all other players. The function $g_n^t\left(s_n^t, \left(a_n^t, a_{-n}^t\right)\right)$ has been defined in \eqref{eqn:costFunction} Section~\ref{sec:utility} capturing the costs to the $n$th player at the $t$th stage. Note that the utility function depends only on the initial state variable $s_n^0$, since the subsequent states $s_n^t$ are determined by \eqref{eqn:stateTransistion}. The function 
		\begin{equation}
			g_n^T\left(s_n^T\right) = s_n^T
			\label{eqn:gT}
		\end{equation}				 
		 expresses a penalty for the $n$th player that is incurred by ending up in state $s_n^T$, i.e.~its SOC, at the end of the scheduling period.
	\end{longtable}%
}
\addtocounter{table}{-1}
\end{definition}
We represent the decision problem of the $n$th player as the following optimisation problem:
\begin{empheq}[box=\fbox]{equation}
\label{eqn:dynamicGame}
\begin{split}
  G_n\left(a_{-n}\right) \hspace{3cm}	& \makebox[0pt][l]{\text{given }}\phantom{\text{subject to }}\ s^0\in\mathcal{S} \\
  										& \makebox[0pt][l]{$\underset{a_n}{\text{\reviewDel{minimise }\reviewAdd{maximise }}}$}\phantom{\text{subject to }}\ U_n\left(s_n^0,\left(a_n, a_{-n}\right)\right)\\
  										& \text{subject to }\ a_n^t\in\mathcal{H}_n^t\left(s_n^t\right)\\
  										& \hphantom{\text{subject to }}\ s_n^{t+1} = f_n^t\left(s_n^t, a_n^t\right)
\end{split}
\end{empheq}
Moreover, the game is referred to as $\left\{ G_1,G_2,\dots,G_N\right\}$\reviewAdd{.}
\begin{definition}
	A strategy profile $\hat{a}=\left(\hat{a}_1,\dots,\hat{a}_N\right)$ is a \textit{Nash equilibrium} for the game $\left\{ G_1,\dots,G_N\right\}$ if and only if for all players $n\in\mathcal{N}$ we have
	\begin{equation}
		U_n\left(s_n^0,\left(\hat{a}_n,\hat{a}_{-n}\right)\right) \reviewDel{\leq }\reviewAdd{\geq }U_n\left(s_n^0,\left(a_n,\hat{a}_{-n}\right)\right),\ \ \forall a_n\in\mathcal{H}_n\left(s_n^0\right)\ .
	\end{equation}
\end{definition}

\subsection{Analysis of the Game}
In order to analyse the game $\left\{G_1,\dots,G_N\right\}$, we follow the dynamic programming (DP) idea by Nie\etal~\cite{Nie2006}. To do so, we introduce notation for subproblems of \eqref{eqn:dynamicGame}. Furthermore, we show an important result about Nash equilibria for these subproblems, which constitutes the basis for the DP-algorithm. Applying the general algorithm eventually leads us to an analytic formulation of the $n$th player's best response $\hat{a}_n$, given the strategies $a_{-n}$ of other players at stage $t$ of a $T$-stage game.

\subsubsection{Subgame Formulation}
For subproblems that are only interested in decisions taken from stage $t\reviewAdd{'}$ onwards, we write:
\begin{gather*}
	s_n^{t\reviewAdd{'},T-1}:=\left(s_n^{t\reviewAdd{'}},\dots,s_n^{T-1}\right),\ \ s^{t\reviewAdd{'},T-1}:=\left(s^{t\reviewAdd{'}},\dots,s^{T-1}\right)\\
	a_n^{t\reviewAdd{'},T-1}:=\left(a_n^{t\reviewAdd{'}},\dots,a_n^{T-1}\right),\ \ a^{t\reviewAdd{'},T-1}:=\left(a^{t\reviewAdd{'}},\dots,a^{T-1}\right)\\
	U_n^{T-t\reviewAdd{'}}\left(s_n^{t\reviewAdd{'}}, \left(a_n^{t\reviewAdd{'},T-1}, a_{-n}^{t\reviewAdd{'},T-1}\right)\right)=\reviewAdd{-}g_n^T\left(s_n^T\right) \reviewDel{+ }\reviewAdd{- }\sum_{\tau=t\reviewAdd{'}}^{T-1}g_n^\tau\left(s_n^\tau, \left(a_n^\tau, a_{-n}^\tau\right)\right)\\
	\mathcal{H}_n^{t\reviewAdd{'},T-1}\left(s_n^{t\reviewAdd{'}}\right):= \left\{a_n^{t\reviewAdd{'},T-1}~|~h_n^\tau\left(s_n^\tau,a_n^\tau\right)\leq 0,\ \tau=t\reviewAdd{'},t\reviewAdd{'}+1,\dots,T-1\right\}\ .
\end{gather*}
For $t\reviewAdd{'}\in\mathcal{T}$ we define a subproblem of the $n$th player as the following optimisation problem:
\begin{empheq}[box=\fbox]{equation}
\label{eqn:dynamicSubGame}
\begin{split}
  G_n^{T-t\reviewAdd{'}}\left(a_{-n}^{t\reviewAdd{'},T-1}\right) \hspace{2cm}	& \makebox[0pt][l]{\text{given }}\phantom{\text{subject to }}\ s^{t\reviewAdd{'}}\in\mathcal{S} \\
  										& \makebox[0pt][l]{$\underset{a_n}{\text{\reviewDel{minimise }\reviewAdd{maximise }}}$}\phantom{\text{subject to }}\ U_n^{T-t\reviewAdd{'}}\left(s_n^{t\reviewAdd{'}},\left(a_n^{t\reviewAdd{'},T-1},a_{-n}^{t\reviewAdd{'},T-1}\right)\right)\\
  										& \text{subject to }\ a_n^{t\reviewAdd{'},T-1}\in\mathcal{H}_n^{t\reviewAdd{'},T-1}\left(s_n^{t\reviewAdd{'}}\right)\\
  										& \hphantom{\text{subject to }}\ s_n^{t\reviewAdd{'}+1} = f_n^{t\reviewAdd{'}}\left(s_n^{t\reviewAdd{'}}, a_n^{t\reviewAdd{'}}\right)
\end{split}
\end{empheq}
Therefore, the subgame is referred to as $\left\{ G_1^{T-t\reviewAdd{'}},G_2^{T-t\reviewAdd{'}},\dots,G_N^{T-t\reviewAdd{'}}\right\}$.
\begin{theorem}
\label{thm:subgameNE}
	Let $\hat{a}=\left(\hat{a}^{0},\dots,\hat{a}^{T-1}\right)$ constitute a Nash equilibrium for the game $\left\{G_1, \dots, G_N\right\}$ with the corresponding trajectories of states $\hat{s}=\left(\hat{s}^{0},\dots,\hat{s}^{T}\right)$. Consider the subgame $\left\{ G_1^{T-t},\dots,G_N^{T-t}\right\}$ for each $t\in\mathcal{T}$. Then, the truncated strategy $\hat{a}^{t,T-1}=\left(\hat{a}^{t},\hat{a}^{t+1},\dots,\hat{a}^{T-1}\right)$ comprises a Nash equilibrium for the subgame $\left\{ G_1^{T-t},\dots,G_N^{T-t}\right\}$.
\end{theorem}
\begin{proof}
	\reviewDel{original proof}\reviewAdd{The proof can be found in the Appendix~\ref{sec:thmProof}. \qed}
\end{proof}

\subsubsection{The DP-Algorithm and Derivation of the Best Response Solution}
Based on the results of the previous subsection, we can formulate the following DP-algorithm to find the solution to the decision problem $G_n(a_{-n})$ \eqref{eqn:dynamicGame}, i.e.~the optimal decision for the $n$th player given the decisions $a_{-n}$ of the other players.

\LinesNotNumbered
\begin{algorithm}[thb]
\caption{DP algorithm for player $n\in\mathcal{N}$ to find the solution to \eqref{eqn:dynamicGame}.}
\label{alg:DP_bestResponse}
\KwInput{$T$, $a_{-n}$, $s^0$} 
\kern-6pt
\hrulefill\\
$t\leftarrow T$ \\
\nextnr \label{algLN:for_sT}
\For{\normalfont \textbf{each} $s^T\in\mathcal{S}$}{
	$V_n^0(s^T_n) \leftarrow g_n^T\left(s_n^T\right)$\\
\nextnr \label{algLN:while}
	\While{$t>0$}{
		$t\leftarrow t-1$\\
\nextnr	\label{algLN:for_st}
		\For{\normalfont \textbf{each} $s^t\in\mathcal{S}$}{
			\begin{align*}	
				\hat{a}_n^{t}\leftarrow \reviewDel{\argmin}\reviewAdd{\argmax}_{a_n^{t}\in\mathcal{H}_n^{t}\left(s_n^t\right)}\reviewAdd{-}g_n^t\left(s_n^t, \left(a_n^t, a_{-n}^t\right)\right) \reviewDel{+ }\reviewAdd{- }V_n^{T-t-1}\left(f_t\left(s_n^t, a_n^t\right)\right)\\
				V_n^{T-t}(s_n^t) \leftarrow \reviewDel{\min}\reviewAdd{\max}_{a_n^{t}\in\mathcal{H}_n^{t}\left(s_n^t\right)}\reviewAdd{-}g_n^t\left(s_n^t, \left(a_n^t, a_{-n}^t\right)\right) \reviewDel{+ }\reviewAdd{- }V_n^{T-t-1}\left(f_t\left(s_n^t, a_n^t\right)\right)
			\end{align*}
		}
	}
}
\kern-6pt
\hrulefill\\
\KwOutput{$\hat{a}_n$}
\end{algorithm}
Let us apply Algorithm~\ref{alg:DP_bestResponse} to obtain the result to the decision problem $G_n(a_{-n})$ \eqref{eqn:dynamicGame} in closed form. Note that both for-loops (line~\ref{algLN:for_sT} and line~\ref{algLN:for_st}) are treated implicitly by keeping $s^T$ and $s^t$ unspecified throughout the computations. 

Given the total scheduling length $T$, the aggregated decisions $a_{-n}$ of all other players, and the initial SOC $s^0$ of the batteries, at the first step ($t=T$) we set $V^0_n(s^T_n)=s^T_n$ according to \eqref{eqn:gT}. With this we enter the while-loop (line~\ref{algLN:while}) which overwrites $t$ to now represent $t=T-1$. We solve for the best decision $\hat{a}^{T-1}_n$ by solving the following problem
\begin{equation*}
\begin{split}
	\hat{a}^{T-1}_n = \reviewDel{\argmin}\reviewAdd{\argmax}_{a_n^{T-1}}\  &\overbrace{\reviewAdd{-}c_2\left( d_n^{T-1} + a^{T-1}_n + L_{-n}^{T-1}\right)^2 \reviewDel{+}\reviewAdd{-} c_1\left( d_n^{T-1} + a^{T-1}_n + L_{-n}^{T-1}\right) \reviewDel{+}\reviewAdd{-} c_0}^{g_n^{T-1}} \\ &\reviewDel{+}\reviewAdd{-} \underbrace{\vphantom{\left(d_n^T\right)^2}  s^{T-1}_n\reviewDel{+}\reviewAdd{-}a^{T-1}_n}_{V^0_n}
\end{split}
\end{equation*}
where we made use of the transition equation~\eqref{eqn:stateTransistion} to rewrite $V^0_n$. The solution is computed as 
\begin{equation*}
	\hat{a}^{T-1}_n = -s_n^{T-1}\ ,
\end{equation*}
and subsequently we have
\begin{equation*}
	V^1_n = c_2\left( d_n^{T-1} -s_n^{T-1} + L_{-n}^{T-1}\right)^2 + c_1\left( d_n^{T-1} -s_n^{T-1} + L_{-n}^{T-1}\right) + c_0\ .
\end{equation*}
With this, the first step is done and we again overwrite $t$ to now represent $t=T-2$. In this stage we solve the following problem
\begin{equation*}
	\begin{split}
		\hat{a}^{T-2}_n = \reviewDel{\argmin}\reviewAdd{\argmax}_{a_n^{T-2}}\ &\reviewAdd{-}g_n^{T-2}\left(s_n^{T-2}, \left(a_n^{T-2}, a_{-n}^{T-2}\right)\right)\\ 
		&\reviewDel{+}\reviewAdd{-} c_2\left( d_n^{T-1} - \left[s_n^{T-2}+a_n^{T-2}\right] + L_{-n}^{T-1}\right)^2\\ &\reviewDel{+}\reviewAdd{-} c_1\left( d_n^{T-1} - \left[s_n^{T-2}+a_n^{T-2}\right] + L_{-n}^{T-1}\right) \reviewDel{+}\reviewAdd{-} c_0\ .
	\end{split}
\end{equation*}
The solution is computed as 
\begin{equation*}
	\hat{a}^{T-2}_n = \frac{1}{2}\left(d_n^{T-1}-d_n^{T-2} - s_n^{T-2} + L_{-n}^{T-1} - L_{-n}^{T-2}  \right)\ ,
\end{equation*}
from which we obtain
\begin{equation*}
\begin{split}
	V^2_n &= \frac{c_2}{2}\left(d_n^{T-1}-d_n^{T-2} - s_n^{T-2} + L_{-n}^{T-1} - L_{-n}^{T-2}  \right)^2 \\
		&+ c_1\left(d_n^{T-1}-d_n^{T-2} - s_n^{T-2} + L_{-n}^{T-1} - L_{-n}^{T-2}  \right) + 2c_0\ ,
\end{split}
\end{equation*}
finalising the second step. This procedure can be done for all subsequent steps. As the equations increase quickly in size, they become infeasible to quote here. Fortunately though, our calculations provided insight into recurring patterns, which all the solutions seem to follow. Eventually, the solution for an arbitrary stage $t$ of the $T$-stage dynamic game can be written as
\begin{empheq}[box=\widefbox]{equation}
\label{eqn:solutionBestResponse}
\begin{split}
	\vphantom{\left[ \sum_{\tau=t+1}^{T^{T^{T^T}}}\right]}\hat{a}^t_n = \frac{1}{T-t}\left[ \sum_{\tau=t+1}^{T-1}\left(d^\tau_n + L^\tau_{-n}\right) - s^t_n - \left(T-t-1\right)\left(d^t_n + L^t_{-n}\right) \right]
\end{split}
\end{empheq}
\reviewAdd{Note that during the derivation the non-linear battery constraints are not strictly considered. Similar to the forecasting errors (cf.~Section~\ref{sec:algoAndSchedule}), these are considered in our simulation when the equilibrium schedules are actually executed.}

\subsection{The Algorithm\reviewAdd{ and Execution of NE schedules}}
\label{sec:algoAndSchedule}
Similar to~\cite{Pilz2017}, we make use of a best-response algorithm (cf. Algorithm~\ref{alg:bestResponse}) to find the solution to the game. 
\begin{algorithm}[thb]
  \caption{Best-response algorithm for finding a pure NE based on \cite{Shoham2009}}
    \label{alg:bestResponse}\KwInput{$T$, $s^0$} 
\kern-6pt
\hrulefill\\
		initialise random strategy profile $a=(a_n,a_{-n})$ \\
		\nextnr\label{alg:BR_while}	
		\While{there exists a player $n$ for whom $a_n$ is not a best response to $a_{-n}$}{
		\nextnr\label{alg:BR_forN}	
		\For{\normalfont \textbf{each} $n\in\mathcal{N}$}{
		\nextnr	\label{alg:BR_forT}	
			\For{\normalfont \textbf{each} $t\in\mathcal{T}$}{
				$\hat{a}_n^t \leftarrow$ best response to $a_{-n}$ based on~\eqref{eqn:solutionBestResponse}
			}	
        	$a_n \leftarrow \left(\hat{a}_n^0,\dots,\hat{a}_n^{T-1} \right)$	
		}		    
	}
\kern-6pt
\hrulefill\\
\KwOutput{$\hat{a}$}
\end{algorithm}
Whereas in~\cite{Pilz2017} an extensive search for optimal schedules $\hat{a}_n$ was performed, here we can compute the best response for each stage (line~\ref{alg:BR_forT}) analytically by means of~\eqref{eqn:solutionBestResponse} and concatenate the results to obtain the optimal schedule $\hat{a}_n$ in response to $a_{-n}$. Performing this computation for each player $n\in\mathcal{N}$ (line~\ref{alg:BR_forN}) results in a new strategy profile $a$. We iterate this (line~\ref{alg:BR_while}) as long as ``there exists a player $n$ for whom $a_n$ is not a best response to $a_{-n}$''. In the actual implementation, this check is done by comparing the current strategy profile with the one obtained from the previous iteration. If it did not change, up to machine precision, an equilibrium is reached and $\hat{a}=\left(\hat{a}_n,\hat{a}_{-n}\right)$ constitutes the Nash equilibrium.

Based on the definition of \reviewDel{the }\reviewAdd{a }NE, no household can benefit from unilaterally deviating from its respective schedule. Nonetheless, we have to keep in mind that it is based on forecasted demand and renewable generation. Whenever either the demand or the generation does not match the forecasted value, it might not be possible anymore to strictly follow \reviewDel{the }\reviewAdd{this }NE schedule. In the analysis in the subsequent sections, we assume that \reviewDel{it }\reviewAdd{every individual }always seeks to be as close as possible to the\reviewAdd{ir} \reviewDel{forecasted }\reviewAdd{determined NE }schedule. To illustrate the idea: Imagine \reviewDel{the }\reviewAdd{a }NE schedule of household $n$ \reviewDel{scheduled }\reviewAdd{requires them }to discharge an amount $x$ in a certain interval. Due to a forecasting error for the renewable generation, this has not been charged fully\reviewAdd{ and can thus not be delivered}. In this case, the schedule will discharge as much as possible during this interval. The deviation from the NE will decrease the benefit in terms of PAR reductions and achieved savings for the consumer. \reviewDel{Nevertheless, }\reviewAdd{Anticipating the results, we want to highlight that }in the following section we show that the solution is robust with respect to these deviations\reviewAdd{ and gives considerable improvements in comparison to other approaches in the literature}.

%% file: sec/4_results_newStructure.tex
In this section, we firstly summarise important simulation parameters and introduce the specific datasets for electricity demand and generation from the photovoltaic (PV) installation. \reviewDel{Secondly, all results are shown with detailed explanations of the individual parameters under investigation. Thirdly, the results are discussed and compared to the literature.
The correctness of the implementation of Algorithm~\ref{alg:bestResponse} is provided in the Appendix~\ref{sec:convResults}.}
\reviewAdd{After analysing the convergence behaviour of the iteration algorithm, }we compare the game--theoretic approach introduced in this manuscript (cf.~Section~\ref{sec:game}) with a simpler non-cooperative static game, revealing the advantages of the dynamic treatment. Subsequently, the analysis of how the participation rate of the DSM scheme and the forecasting errors influence the scheduling outcome is shown. Finally, we consider the influence of the composition of the neighbourhood on the peak-to-average ratio (PAR) reduction.
This is an important measurement of the effectiveness of the DSM scheme. We consider the PAR of the aggregated electricity load \eqref{eqn:totalLoad} over the respective scheduling period. It is defined by
\begin{equation}
	\text{PAR} = T\cdot\frac{\max_{t\in\mathcal{T}}L^t}{\sum_{t\in\mathcal{T}}L^t}\ .
	\label{eqn:PAR}
\end{equation}

\subsection{The Simulation Setup}
\label{sec:simSetup}
In the real-world application, the smart meter of individual households collects data about electricity demand and generation from the available PV installation. As specified in Section~\ref{sec:smartGrid}, the demand-side management (DSM) protocol requires participants to send forecasts of the demand and generation to the utility company. These forecasts are based on historically collected data. In order to run our simulations, we omit this forecasting step and rather make use of two publicly available data sets.

\paragraph{Demand data:}
The demand data stem from the openei dataset~\cite{Openei2013a}. It contains 365 days of simulated hourly data\footnote{We make use of $T=24$ for all simulations, if not stated otherwise.} for households in TMY3-locations in the USA~\cite{Nrel2015}. The building models used for this simulation can be found in~\cite{Openei2013}. Based on an additional survey, all buildings are put into one of three different categor\reviewDel{y}\reviewAdd{ies}. They differ with respect to their overall consumption. Following~\cite{Openei2013a}, we refer to them as LOW, BASE and HIGH consumers. For all simulation runs, we picked the same $M=25$ households, in close vicinity to each other, to represent our neighbourhood. With respect to their consumption categories, we have seven LOW, nine BASE and nine HIGH users. 

\paragraph{PV data:}
Data for the PV generation are based on real-world measurements~\cite{PowerNetworks2014} in the UK. They contain hourly values for days between September 2013 and October 2014. Note that latitude and climate zone of the measurement location are similar to the ones of the demand data. \reviewDel{In }\reviewAdd{Under the }assumption that the weather for all households in the neighbourhood is the same, we use data from the same site for each of them. An estimate for the $kW\!p$ value is obtained from looking at the highest hourly output in the course of a whole year. Its value is $w_{\max}=3.7\;$kWh, which is why we assume $kW\!p\approx 4\;$kW. We account for different sizes of PV installations by scaling the data set with a household specific factor $p_{n}$. About $6\%$ of the collected data was corrupted. We set all these values to $w=0.0\;$kWh. This does not pose any problem for our simulation results, but can be seen as realistic failures of the installation. 

\paragraph{Battery and pricing parameters:} The parameters of the battery are based on the Tesla Powerwall 2~\cite{Tesla2017} data sheet. The choice to employ this battery system is motivated by two reasons: (i) The same battery was used in~\cite{Pilz2017}, allowing for a direct comparison of the results. (ii) A non-extensive analysis of different battery systems showed that the Tesla Powerwall 2 qualifies as a representative of state-of-the-art technology. Please see the Appendix~\ref{sec:batteryJust} for more details. A summary of the battery parameters can be found in Table~\ref{tab:batPara}. The data sheet only specifies the round-trip efficiency $\eta = \eta^+\cdot\eta^-$ of the battery. Without loss of generality, we assume that charging and discharging contribute equally, yielding $\eta^+ = \eta^- = \sqrt{0.918}$. 

\begin{table}[t]
	\caption{\textit{Battery parameters.} Parameters for a Tesla-inspired~\cite{Tesla2017} home battery storage system.}
	\centering
	\label{tab:batPara}
	\begin{tabular}{cc}
		\toprule
		\textbf{Variable}	& \textbf{Value}	\\
		\midrule
    	$\eta^+$ & $0.958$\\ 	
    	$\eta^-$ & $0.958$\\
    	$\eta_{\text{inv}}$ & $0.960$\\
		$\rho^+$ & $5.0$\;kW/h\\   
    	$\rho^-$ & $-7.0\;$kW/h\\
		$\bar{\rho}$ & $-0.001$\\
    	$s_{\max}$ & $13.5$\;kWh \\
    	$s_{\min}$ & $0.0$\;kWh\\
    	$s^*$ & $9.46$\;kWh\\
		\bottomrule
	\end{tabular}
\end{table}
For the parameters in the cost function~\eqref{eqn:costFunction} we use $c_2 = 0.03125\;$\$/MW$^2$, $c_1=1.0\;$\$/MW, and $c_0=0$, following other studies~\cite{Pilz2017,Rahbar2015}. This allows to \reviewAdd{directly }compare our results. 

\subsection{Convergence Behaviour of the Algorithm}
\label{sec:convResults}
Let us provide an insight into the convergence behaviour of Algorithm~\ref{alg:bestResponse}. 
\begin{figure}[thb]
	\centering
  	\includegraphics[width=0.79\textwidth]{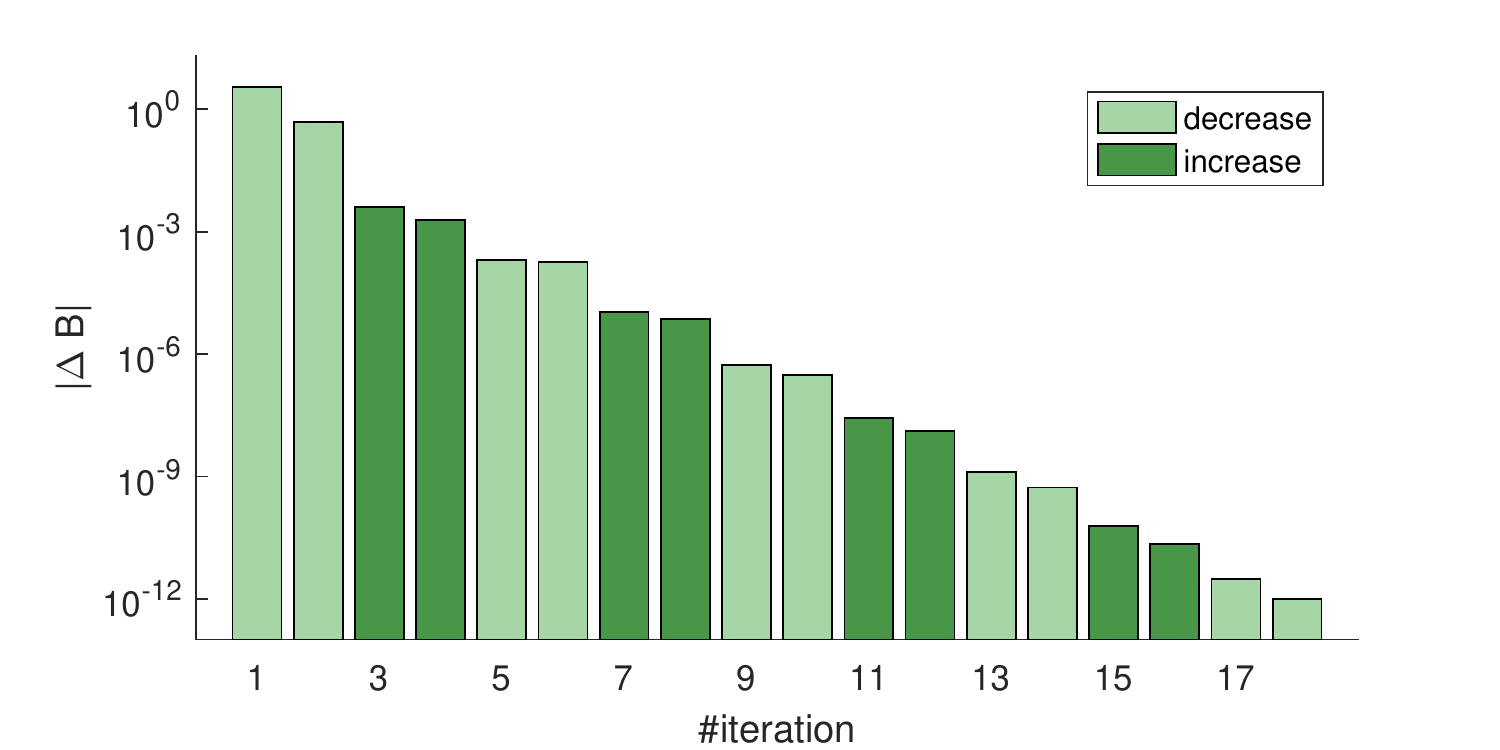}
  	\caption{\textit{Convergence analysis of Algorithm~\ref{alg:bestResponse}.} The change of the average bill of the DSM participants between consecutive iterations is plotted over the iteration number. The ordinate is scaled logarithmically. Bars are coloured according to the algebraic sign of the change. The specific data points stem from a simulation in Section~\ref{sec:parRate_Error} with $64\%$ participation rate and without forecasting errors.}
  	\label{fig:utilIter}
\end{figure}
The condition that needs to be fulfilled to declare equilibrium is stated as `there exists no player $n$ for whom his current action $a_n$ is not a best response to the actions $a_{-n}$ of the other players' (cf.~Algorithm~\ref{alg:bestResponse}, line~\ref{alg:BR_while}). \reviewAdd{Within our specific implementation of \texttt{selma}, the stopping criteria is based on the L2 difference between the action profiles of two consecutive iterations, i.e.~when this difference is smaller or equal to $10^{-15}$ the algorithm breaks out of the loop.}
\begin{figure}[t]
	\centering
  	\includegraphics[width=0.69\textwidth]{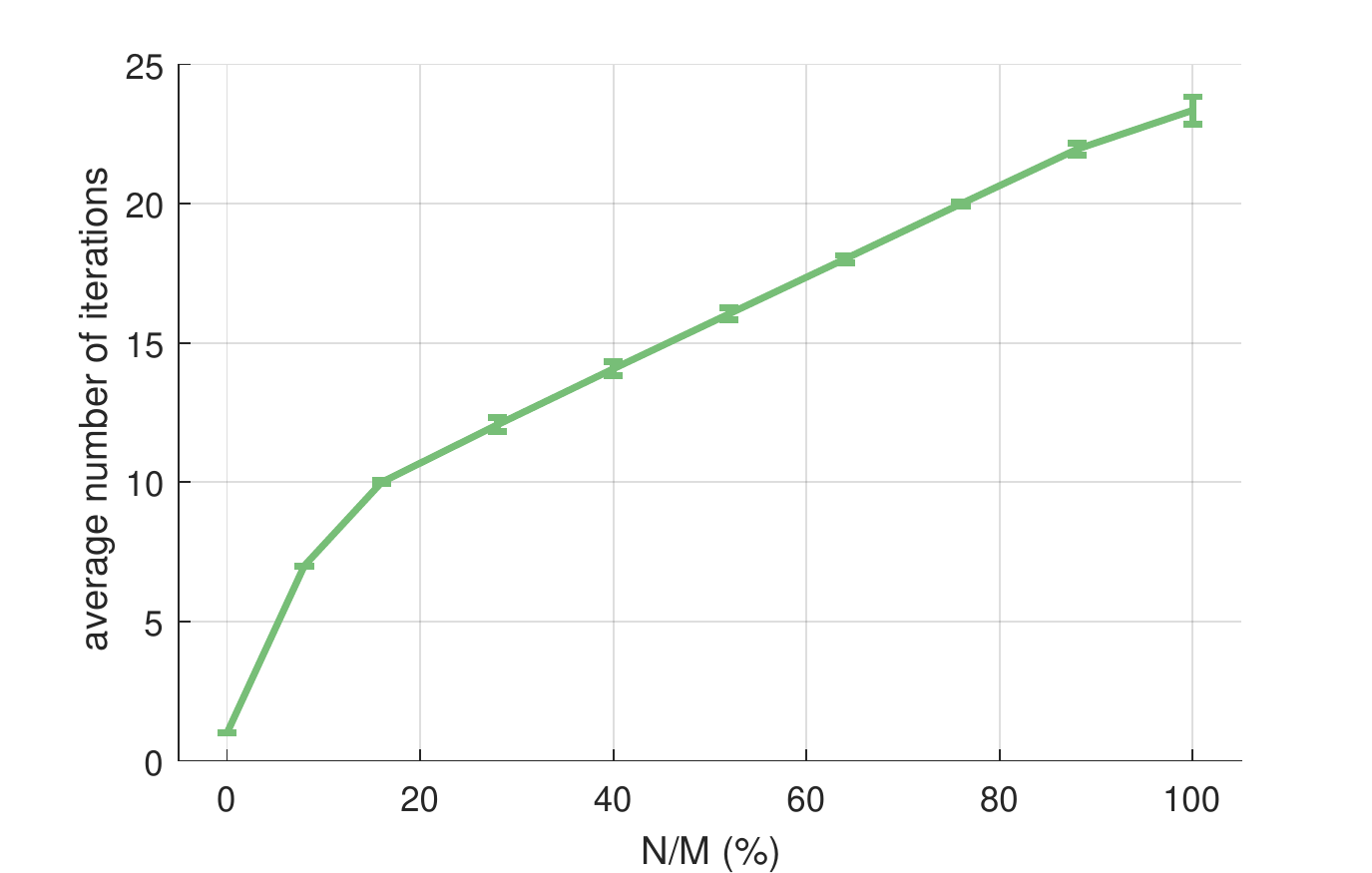}
  	\caption{\textit{Iteration statistics.} The mean number of iterations per day is plotted over the participation rate in per cent. The values stem from the simulations undertaken in Section~\ref{sec:parRate_Error}\reviewDel{, where both the case with and without forecasting error are considered}. In addition to the average over 365 days, the standard deviation is shown for each data point. \reviewDel{The dotted line corresponds to the right-hand axis and shows the difference between the two curves in per cent.}}
  	\label{fig:iterStat}
\end{figure}
Associated with the current action profile during each iteration are also the energy bills for each participant. 
\paragraph{Results:} In Figure~\ref{fig:utilIter}, the absolute change of the average bill $B=\nicefrac{1}{N}\sum_{n}B_n$ (cf.~\eqref{eqn:billDSM}) is shown for a randomly selected day of the simulation shown in Section~\ref{sec:parRate_Error}. To cover the large scale of different changes, a logarithmic representation is chosen. The respective sign of the change is then expressed in the colour of the bar. 

Figure~\ref{fig:iterStat} shows how the number of average iterations per day depends on the number of participants in the DSM scheme. \reviewDel{Furthermore, it reveals the influence of the error on the iteration statistics. }The values are again taken from the simulations in Section~\ref{sec:parRate_Error}.

\paragraph{Discussion:}
\label{sec:discConvergence}
The results give evidence of a correctly working iteration algorithm (cf.~Algorithm~\ref{alg:bestResponse}). From Figure~\ref{fig:utilIter} we see that between any two consecutive iterations, the absolute change of the average electricity bill is monotonically decreasing. Furthermore, we observe that the rate of this decrease is almost linear in the semi-logarithmic plot, hinting towards an exponential relationship. 

Due to the exponential convergence towards \reviewDel{the }\reviewAdd{a }Nash equilibrium, only few iterations are needed to obtain the equilibrium schedules. The specific number of iterations depends on the number of participants taking part in the DSM scheme. This is comparable to the ones shown in~\cite{Nguyen2015}. Figure~\ref{fig:iterStat} shows that the average number of iterations increases monotonically with the number of participants. Moreover, the variation across the number of iterations for individual scheduling periods is small, as shown by the standard deviation. This is a strong result, as it shows that the convergence properties are insensitive to different demand data of the individual participants. \reviewDel{Additionally, when comparing the number of iterations for the two scenarios: (i) with forecasting error for demand and generation and (ii) without any forecasting error, it becomes clear that the convergence behaviour is not influenced. The biggest difference is observed for $100\%$ participation rate, where the difference amounts to approximately $0.5\%$.} \reviewAdd{During experimentations with the code, more than one million games were solved which all converged to a Nash equilibrium.}

The small number of iterations directly translates to small computational times and thus does not hinder a real-world application. Typical 365-day simulation runs take about $30\;$s on a single core of an \texttt{i7-3770S} CPU and require less than 1\;GB of memory. Note that in the real-life scenario, the scheduling process is initiated once before the scheduling period and only needs to calculate the equilibrium schedules for the upcoming day. In summary, we expect no difficulties in implementing a DSM scheme based on our scheduling software \texttt{selma}.

\subsection{Comparison Between a Static and a Dynamic DSM scheme}
\label{sec:resultsComparison}
In~\cite{Pilz2017}, a similar DSM scheme to the one described in Section~\ref{sec:smartGrid} was examined. 
Both are based on a battery scheduling game for households of a neighbourhood served by the same utility company (UC). Their main difference is the underlying game that determines the schedules for the upcoming day. Whereas in this paper we employ a discrete time dynamic game,~\cite{Pilz2017} made use of a simpler non-cooperative static game in which players were only able to choose between four discrete options for each interval. For a more thorough description please see~\cite{Pilz2017}. For the sake of comparison, none of the households is equipped with PV cells.
\begin{figure}[t]
	\centering 
	\subfigure[week 12\label{fig:week12load}]{\includegraphics[trim=0 0 0 0,clip,width=0.81\linewidth]{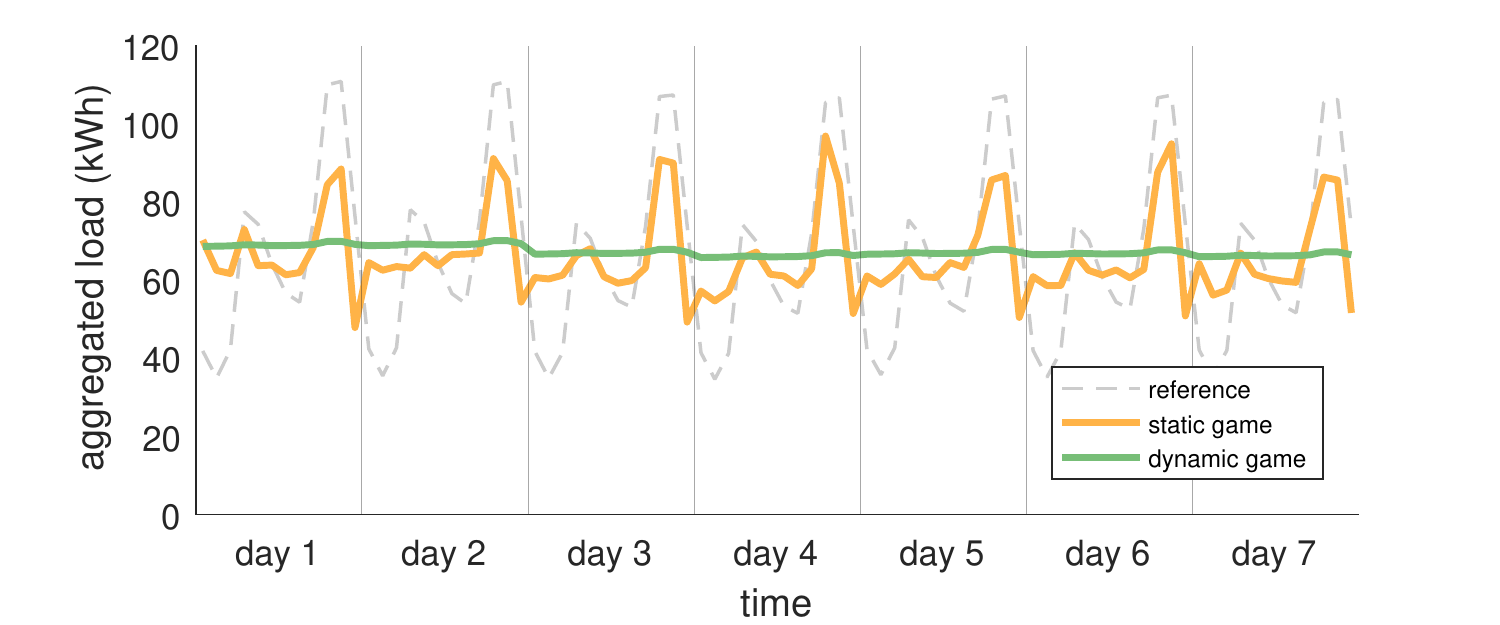}}
	\subfigure[week 38\label{fig:week38load}]{\includegraphics[trim=0 0 0 0,clip,width=0.81\linewidth]{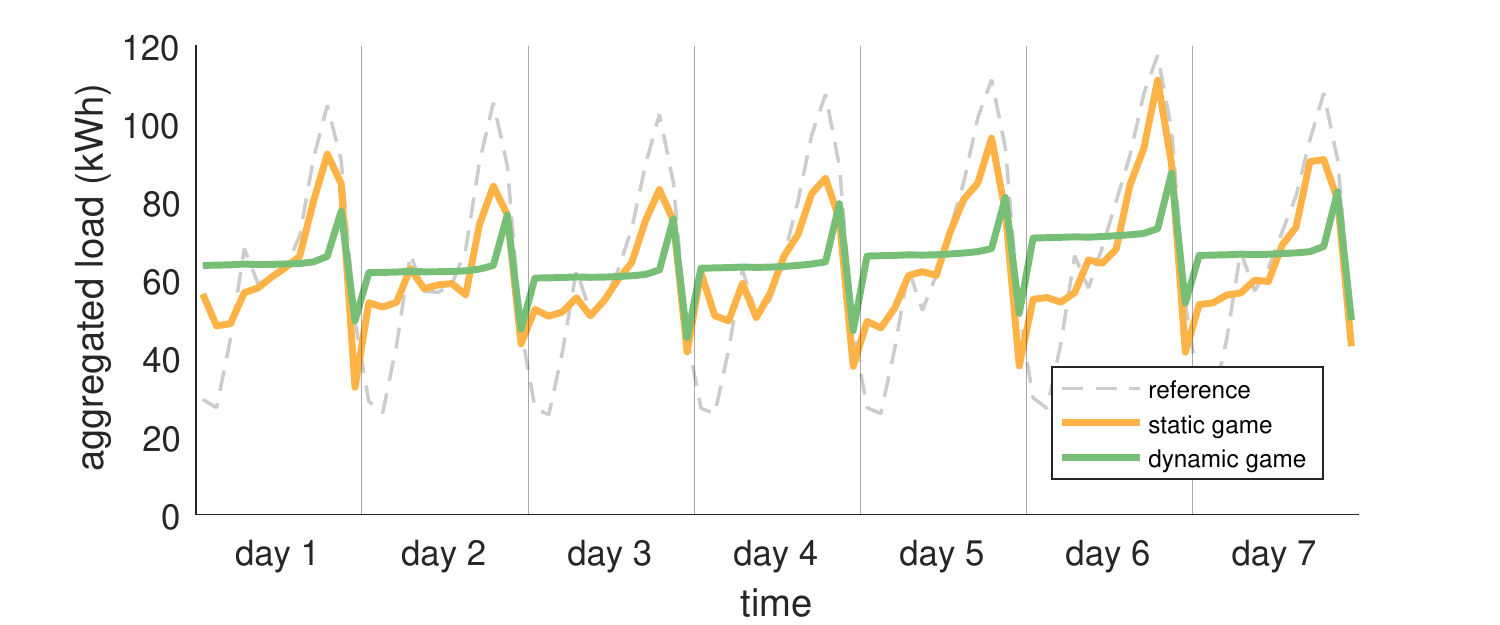}}
	\vspace{-0.5\baselineskip}
	\caption{\textit{Load comparison.} The aggregated load of all households in the neighbourhood is plotted over time. The simulation domains cover seven days, that were each scheduled consecutively. The aggregated demand is given as a reference. The orange curve results from a DSM scheme employing a static scheduling game~\cite{Pilz2017}, while the green one stems from a DSM scheme employing a dynamic game. Other than the underlying game structure, all parameters are identical.}
	\label{fig:LoadComparison}
\end{figure}
\begin{table}
\caption{\textit{PAR comparison.} Peak-to-average ratios calculated as the average over the individual days of week~12, week~25, week~38, and week~51 for the case without storage system (Reference) and both underlying games of the DSM scheme. $\mu$ gives the average over all four weeks. Static: game employed in~\cite{Pilz2017}; Dynamic: game described in Section~\ref{sec:game}. The values in parentheses represent the standard deviation.}
	\label{tab:comparePAR}
	\centering
	\begin{tabular}{ccC{2.3cm}C{2.3cm}C{2.3cm}}
		\toprule
		&	& \textbf{Reference} & \textbf{Static} & \textbf{Dynamic} \\
		\midrule
    	\multirow{4}{*}{\textbf{Period}} & week 12 & $1.623$ $(0.005)$ & $1.374$ $(0.070)$ & $1.013$ $(<\!0.001)$ \\
    	&week 25 & $1.574$ $(0.033)$ & $1.410$ $(0.035)$ & $1.198$ $(0.016)$ \\
    	&week 38 & $1.685$ $(0.031)$ & $1.439$ $(0.080)$ & $1.231$ $(0.015)$ \\
    	&week 51 & $1.718$ $(0.037)$ & $1.468$ $(0.082)$& $1.015$ $(0.001)$ \\ \midrule
    	&$\mu$ & $1.650$ $(0.064)$ & $1.423$ $(0.040)$ & $1.114$ $(0.117)$\\
	\bottomrule
	\end{tabular}
\end{table}

In this subsection, we compare the two approaches with respect to their success in reducing the PAR of the aggregated load. To this end, the same parameters for each household and also the same demand data are used. Households do not have the capability of on-site generation, but are equipped with the same batteries (cf.~Table~\ref{tab:batPara}). The upcoming day is divided into $T=12$ intervals and we assume $N=M=25$, i.e.~every household takes part in the DSM scheme. As in~\cite{Pilz2017}, we simulate full weeks by using the state-of-charge (SOC) values of the batteries at the end of the scheduling period as the initial configuration for the following one. 
\paragraph{Results:} Figure~\ref{fig:week12load} and Figure~\ref{fig:week38load} show the aggregated load curves achieved by the DSM schemes for forecasts given by week 12 and week 38 of the demand data set~\cite{Openei2013a}, respectively.
For completion, we also simulated week~25 and week~51 as done in~\cite{Pilz2017}. A summary of the achieved results can be seen in Table~\ref{tab:comparePAR}.
\begin{figure}[t]
	\centering 
	\subfigure[static scheduling game\label{fig:comp_discrete}]{\includegraphics[trim=0 0  0 0,clip,width=0.45\linewidth]{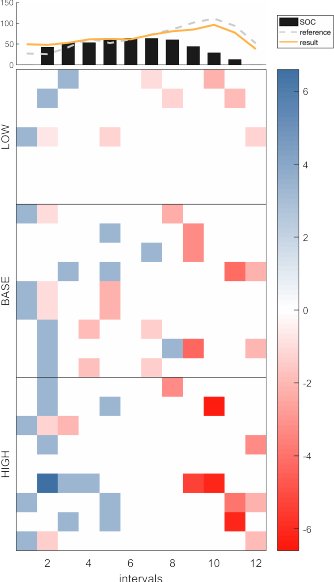}}
	\subfigure[dynamic scheduling game\label{fig:comp_dynamic}]{\includegraphics[trim=0 0 0 0,clip,width=0.45\linewidth]{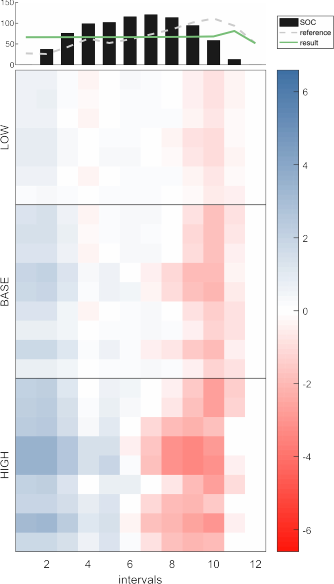}}
	\vspace{-0.5\baselineskip}
	\caption{\textit{Nash equilibrium schedule comparison.} The schedules for all participating households of the demand-side management scheme for a single scheduling period (week 38, day 5, cf.~\cite{Pilz2017}) are shown together with their respective aggregated load and aggregated state-of-charge (SOC). (a) The underlying game structure is a static non-cooperative game from~\cite{Pilz2017}. Within each interval, players can choose between four discrete decisions. (b) Here, the game structure is the dynamic game introduced in Section~\ref{sec:game}. Note that the schedules employ the same scaling.}
	\label{fig:comparison}
\end{figure}

On average, a $14\%$ and a $32\%$ decrease of the PAR value was achieved by the static and the dynamic games, respectively. To understand the differences of the outcomes, we explicitly look at the schedules that are obtained in the NE of the respective games. Figure~\ref{fig:comparison} shows these schedules exemplarily for day 5 of week 38 (Figure~\ref{fig:week38load}, cf.~Figure 3 in~\cite{Pilz2017}) together with the aggregated load and aggregated SOC above it. Each row illustrates the equilibrium schedule of one household. 

\paragraph{Discussion:}
Comparing the aggregated load curves (cf.~Figure~\ref{fig:LoadComparison}) shows that a DSM scheme based on a dynamic game can achieve an almost flat profile. Nevertheless, depending on the given data, the outcome of the scheduling is subject to a finite-horizon effect. Empirically, we observe peaks and troughs at the end of the scheduling period if the demand for the final interval is lower than the average demand of the whole day. This indicates that the starting time of the DSM scheme has an influence on the achievable outcome. Nonetheless, this parameter is fixed through the DSM scheme protocol, thus asking for alternative solutions to the finite-horizon effect. Future work will aim to eliminate the influence of the starting time altogether.

In Table~\ref{tab:comparePAR}, we observe that on average the dynamic game reduces the PAR value more than twice as much as the static game. However, with respect to the individual weeks the static game shows a smaller standard deviation of 0.04 and thus seems to be more consistent. Its achieved reductions are all between $10.4\%$ -- $15.3\%$, while the range of reductions by the DMS scheme with the dynamic game is $23.9\%$ -- $40.9\%$. The differences with respect to the standard deviations is again owed to the finite-horizon effect. It is also present in the case with the static game, but due to generally worse outcome, does not alter it as much as the results of the dynamic scheduling game.

We can further understand the differences between the static and dynamic game from Figure~\ref{fig:comparison}. The restriction to four discrete options for each interval in the static case, i.e.~(i) remain idle, (ii) charge half interval, (iii) charge full interval, and (iv) use battery, results in a majority of intervals where the battery remains idle. This is because of a lack of incentive to charge the battery by the two given amounts. In the dynamic game, players can choose to charge their battery from a continuous spectrum of decisions in a given interval. This difference becomes most apparent when looking at the aggregated SOC of all participants. Whereas the maximal SOC in the static case is approximately $64\;$kWh, almost twice as much ($120\;$kWh) is charged in the dynamic case. In summary, it shows that the increased flexibility of the dynamic game is better suited to minimise the PAR of the aggregated load.

\reviewAdd{Note that all these comparisons allow for strong conclusions as they are based on the identical data set~\cite{Openei2013a} and also all the other parameters, such as number of players $N$, number of time intervals $T$, etc. are chosen to be the same. Nevertheless, comparisons to other results in the literature are possible: Compared to the work by Nguyen\etal~\cite{Nguyen2015}, a better PAR reduction is achieved while also the number of iterations to obtain the equilibrium solution is lower by two orders of magnitude. Similarly, the PAR reduction of Yaagoubi\etal~\cite{Yaagoubi2015b} is worse than the approach shown in this manuscript. As they schedule not only the battery but also shift other household appliances, a comparison of the computational costs is not appropriate.}

\subsection{Influence of Participation Rate and Forecasting Errors}
\label{sec:parRate_Error}
The question of how many participants are needed to obtain considerable gains in terms of PAR reduction and savings is important. 
\begin{figure}[th]
	\centering
  	\includegraphics[width=0.79\textwidth]{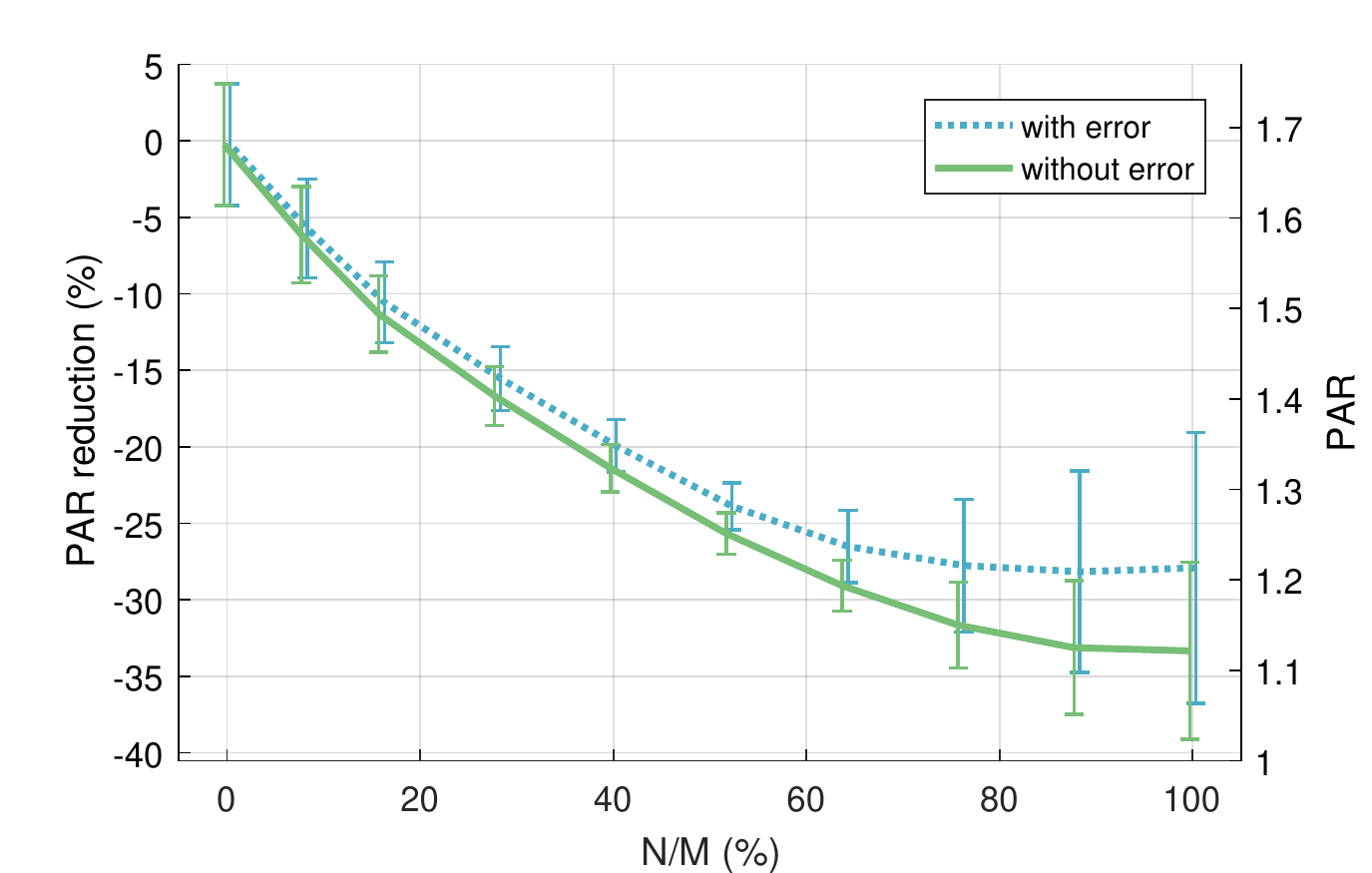}
  	\caption{\textit{Peak-to-average ratio (PAR) reduction dependency on the participation rate.} The mean PAR reduction in per cent is plotted over the participation rate in per cent. The right-hand axis shows the absolute values of the PAR. In addition to the average over 365 days, the standard deviation is shown for each data point. The simulations were run for a scenario with forecasting errors and one without forecasting errors. Note that the data points are slightly shifted along the abscissa to increase readability.}
  	\label{fig:PARvsNM}
\end{figure}
\begin{figure}
	\centering
  	\includegraphics[width=0.79\textwidth]{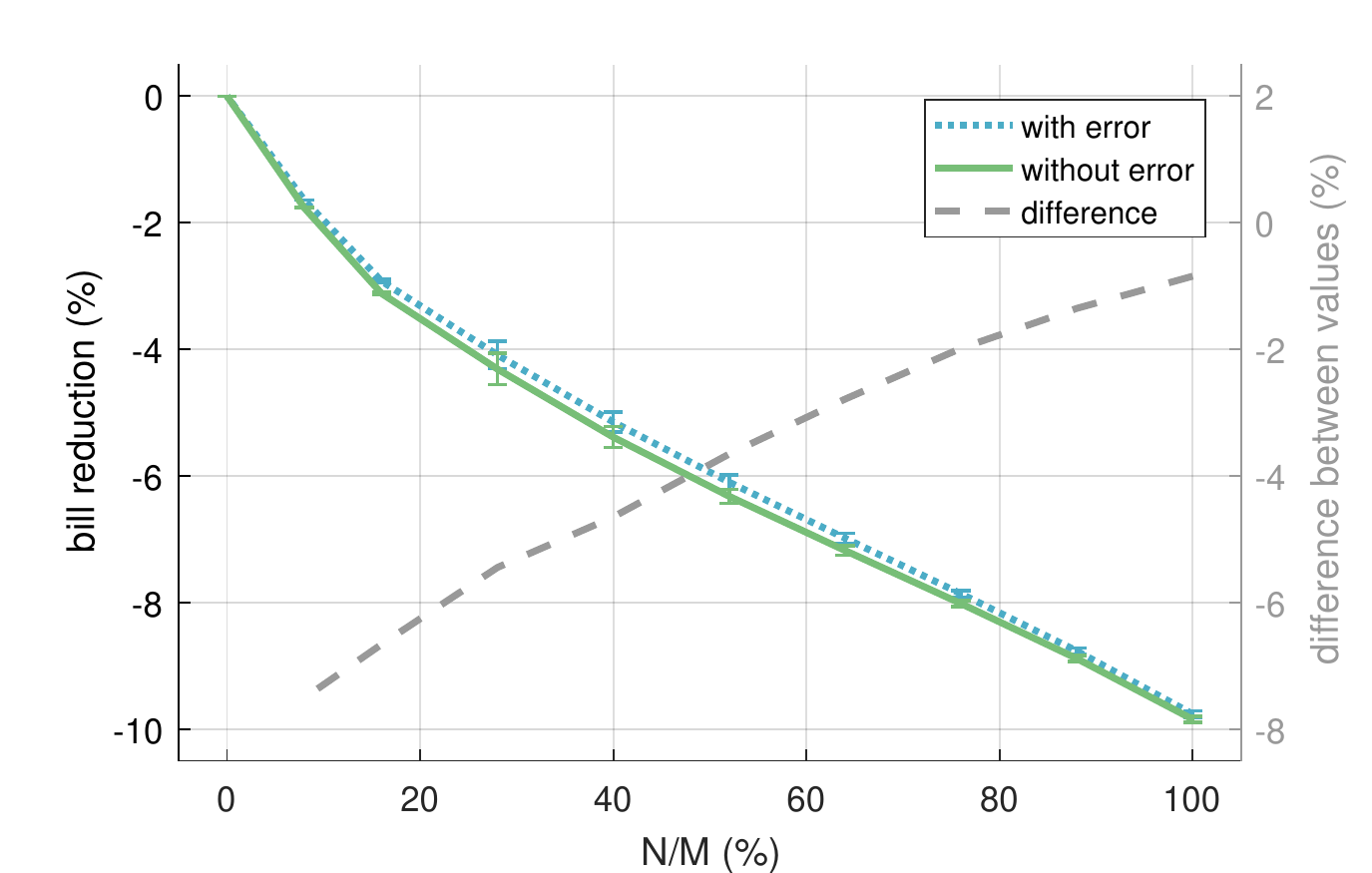}
  	\caption{\textit{Savings dependency on the participation rate.} The mean bill reduction in per cent for participants of the demand-side management scheme are plotted over the participation rate in per cent. In addition to the average over 365 days and participants, the standard deviation between different participants is shown for each data point. The simulations were run for a scenario with forecasting errors and one without forecasting errors. The difference between the two curves is plotted against the right-hand axis.}
  	\label{fig:savingVsNM}
\end{figure}
Moreover, within this subsection the robustness with respect to the forecasting errors (cf.~Section~\ref{sec:error}) is shown. To do so, we assume the forecasting error for the demand to be $\epsilon_d = 8\%$ for every household~\cite{Bichpuriya2016}\reviewAdd{, which could be obtained from a forecast performed by an artificial neural network, and is approximately 2.5 times higher than the best forecast obtained by them}. This is independent of whether the household participates in the DSM scheme or not. The forecasting error for the solar generation is set to $\epsilon_w = 10\%$ in accordance with~\cite{Dolara2015,Rana2016}. \reviewAdd{Rana\etal~\cite{Rana2016} make use of a neural network and clustering of weather data to forecast half hourly solar power output for the upcoming day.} Note that only participants of the DSM scheme are equipped with PV cells and thus subject to the forecasting error. The values are taken to represent a worst-case scenario. Subsequently, any real-world scheduling result should fall in the interval between the worst-case outcome and the respective outcome without any forecasting error.

We simulate a full year and average over the obtained PAR values for the individual days. All participants are equipped with a lithium-ion battery (cf.~Table~\ref{tab:batPara}) and a solar cell. The size of the PV installation depends on the user's category. For LOW, BASE, and HIGH consumers, we use $p_n=0.3$, $p_n=0.5$, and $p_n=0.7$, respectively. Starting with all 25 households taking part in the DSM scheme, we eliminated three users, i.e.~one randomly selected from each consumer category, in each subsequent run. Non-participant still exhibit the specified forecasting error for their demand. 

\paragraph{Results:} Figure~\ref{fig:PARvsNM} shows the reduction of the PAR value over the rate of participating consumers for the scenarios with and without forecasting errors. It includes not only the mean values, but also the standard deviation. Note that we slightly shifted the results for both runs along the abscissa to increase readability. An additional axis on the left indicates the absolute PAR values.
Whereas the PAR reduction is the interest of the UC, the financial rewards, i.e.~savings off the energy bill, are the interests of the participants of the DSM scheme. Figure~\ref{fig:savingVsNM} shows the average saving per day for all participants both with and without forecasting error. For further insight, it also illustrates the difference between the two curves. 

\paragraph{Discussion:} Although a worst-case scenario is simulated, the outcome with respect to PAR reduction (cf.~Figure~\ref{fig:PARvsNM}) and electricity bill (cf.~Figure~\ref{fig:savingVsNM}) reduction show considerable gains for the UC and the participants of the DSM scheme. 

Without forecasting error the PAR reduction monotonically improves with the proportion of the participants. This stands in contrast to the results shown in~\cite{Soliman2014}, where a minimum is reached at medium range participation rate. In comparison to other studies, such as~\cite{Nguyen2015,Longe2017}, we conclude that our dynamic game performs as good as their respective scheduling approach. At 100$\%$ participation rate, a reduction of $-33.3\%\ (5.8\%)$ is achieved, in agreement with the results shown in Section~\ref{sec:resultsComparison}. It should be noted that a perfectly flat load profile corresponds to an approximately $-40\%$ reduction of the PAR. Thus the outcome is close to the theoretical optimum. When looking at the standard deviation, we observe that it is lowest for the simulation run with 52$\%$ participation rate and increases towards both ends of the spectrum. On the lower end of participation rate the fluctuations of the PAR value for different days is just an artefact of the data set in use. Small numbers of participants have not enough influence on the overall neighbourhood to change this. When regarding large participation rates, the PAR value is considerably reduced. The increase of the standard variation for these runs stem directly from the finite-horizon effect already discussed in Section~\ref{sec:resultsComparison}. 

\begin{figure}[thb]
	\centering
  	\includegraphics[width=0.79\textwidth]{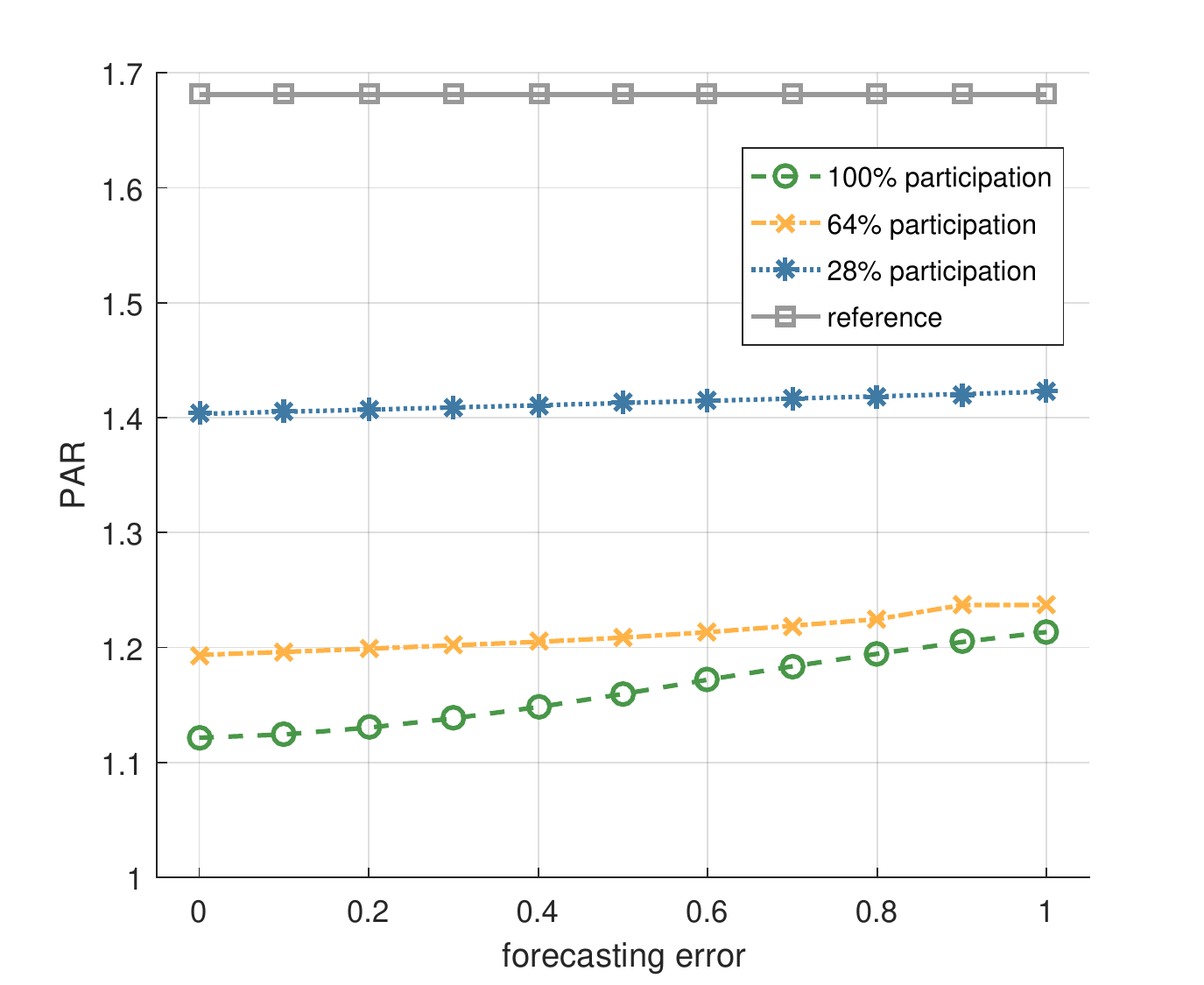}
  	\reviewAdd{\caption{\textit{Influence of forecasting error on the achieved PAR value.} The PAR values for three different participation rates are plotted over the forecasting error. Here the abscissa refers to the worst-case scenario as described in Section~\ref{sec:parRate_Error}, e.g. at 0.6 a forecasting error with the magnitude of 60$\%$ of the worst-case is assumed. In the reference case no game is played, so the forecasting error does not influence the PAR value.}}
  	\label{fig:PARvsError}
\end{figure}

The results for runs with forecasting errors follow the results without errors closely. \reviewAdd{Figure~\ref{fig:PARvsError} shows how the forecasting error affects the achieved PAR values for three selected participation rates when transitioning from the perfect forecast to the worst-case scenario as depicted in Figure~\ref{fig:PARvsNM}.} For low participation rates the difference is negligible but starts to increase when more households participate in the DSM scheme. Nevertheless, even in the worst-case scenario, a reduction of $-27.8\%\ (8.9\%)$ is achieved at 100$\%$ participation rate \reviewAdd{(cf.~Figure~\ref{fig:PARvsNM})}. With respect to the standard deviation, we again recognise similarities to the runs without forecasting errors. Smallest variations in the PAR reduction are obtained for participation rates around $50\%$, while we again see increasing variations at high participation rates. Here, the increase is distinctly larger than in the other runs. The reason behind this difference is directly explained by the forecasting error. As more participants join the DSM scheme, the absolute amount of deviation from the actual demand and production is increasing. 

It is worth noting that the result for a participation rate of $76\%$, i.e. a reduction of $-27.7\%\ (4.4\%)$ \reviewAdd{(cf.~Figure~\ref{fig:PARvsNM})}, are very promising from a practical point of view. The UC might not be able to convince everybody to participate in the DSM scheme, but can still gain reductions of the PAR value close to what is achievable at maximum participation.

\reviewAdd{In~\cite{Pilz2018a} it is investigated what happens when no forecast is calculated. Rather than calculating the forecast, the demand data of the current day is used as the input to the dynamic game that determines the schedules for the next day. It turns out that this approach is oversimplified and results in distinctly worse outcomes (cf.~\cite[Figure 5]{Pilz2018a}), i.e.~it shows the importance of a suitable forecasting mechanism.}\\

The savings that participants of the DSM scheme can gain increase monotonically with the share of participants. Furthermore, we observe that the variations between different participants is negligible. This is due to the particular proportional billing scheme employed in the scheme (cf.~Section~\ref{sec:utility}). It ensures fairness in the sense that LOW and HIGH consumers can gain equally by signing up for the DSM scheme. The difference between runs with and without forecasting errors reveals that the forecasting error does not influence the bill reduction to a great extent. Since the two curves are almost non-separable to the unaided eye, the difference is shown in the same plot (cf.~Figure~\ref{fig:savingVsNM}). It becomes clear that the difference is actually decreasing for larger numbers of participants. 

This highlights that the dynamic scheduling game ensures robust and beneficial results for the participants of the DSM scheme, even in the worst-case scenario.

\subsection{Consumer Type Dependency}
\label{sec:resultsConsumer}
\begin{figure}[thb]
	\centering
  	\includegraphics[width=0.79\textwidth]{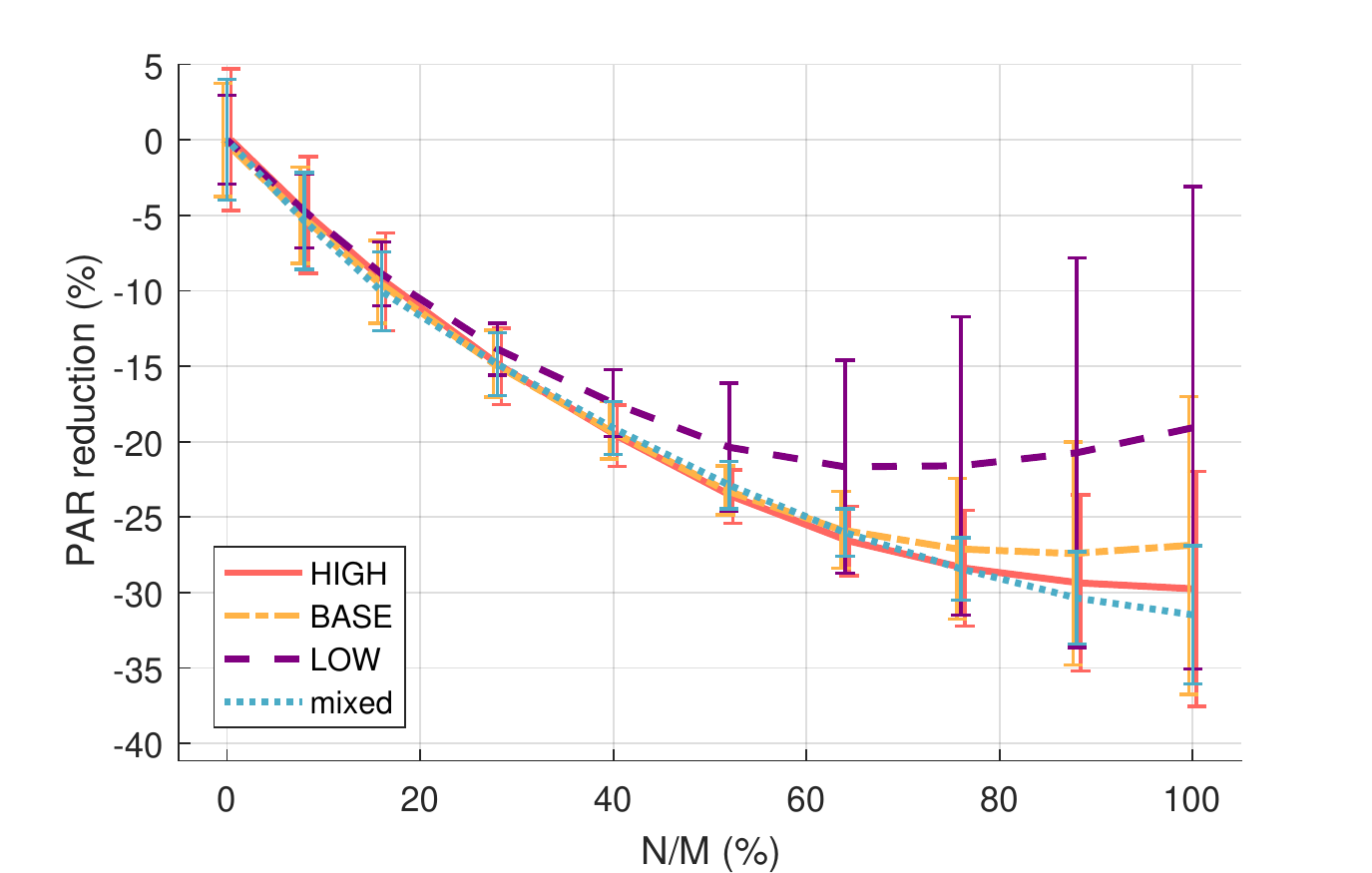}
  	\caption{\textit{Peak-to-average ratio (PAR) reduction for different neighbourhoods.} The mean PAR reduction in per cent is plotted over the participation rate in per cent for different mono-type consumer neighbourhoods. In addition to the average over 365 days, the standard deviation is shown for each data point. For comparison, the results of a mixed neighbourhood (cf.~Figure~\ref{fig:PARvsNM}) are also presented. Note that the data points are slightly shifted along the abscissa to increase readability.}
  	\label{fig:comparison_PARcons}
\end{figure}
\begin{figure}[thb]
	\centering
  	\includegraphics[width=0.79\textwidth]{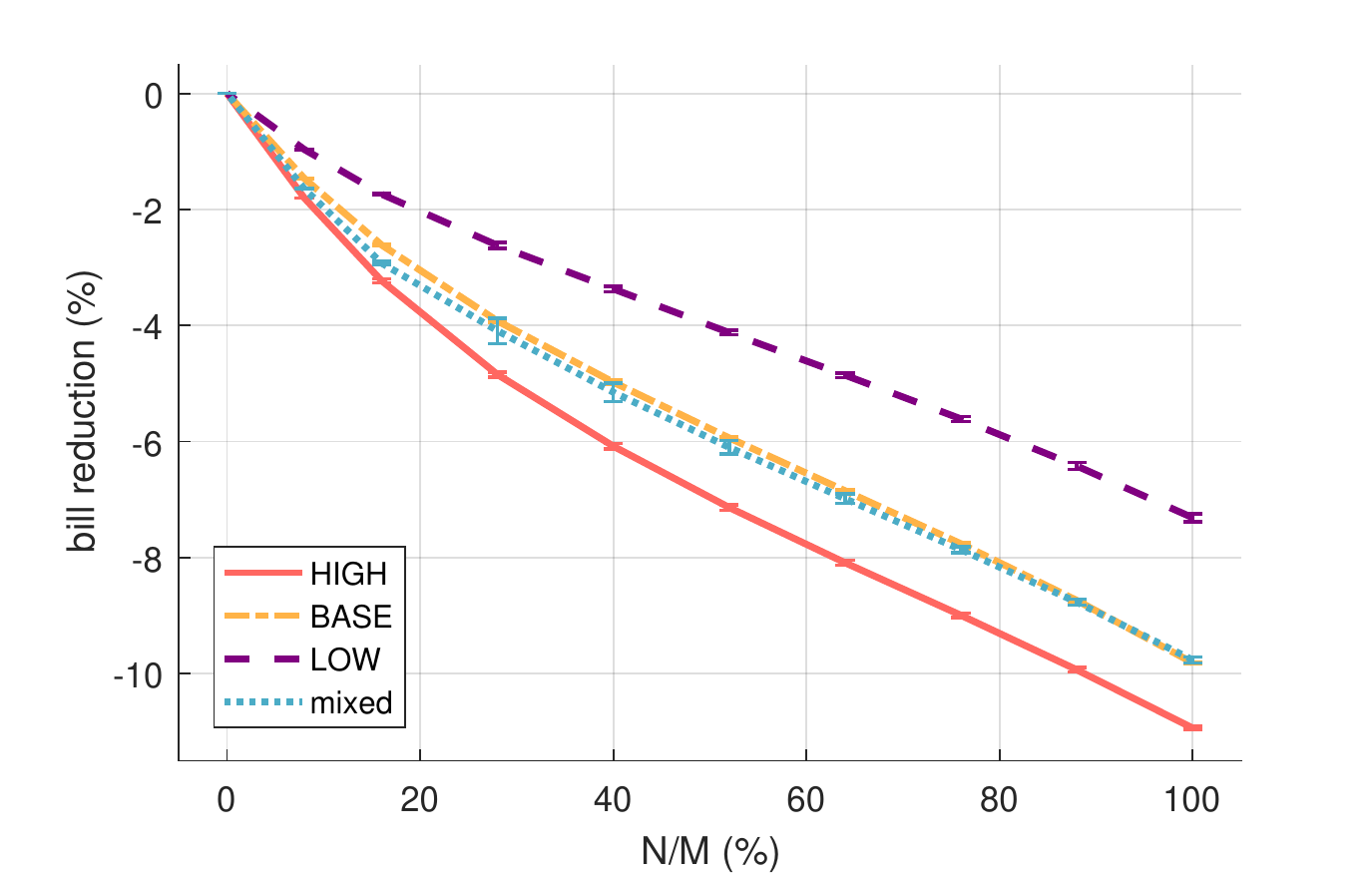}
  	\caption{\textit{Savings for different neighbourhoods.} The mean bill reduction in per cent for participants of the demand-side management scheme are plotted over the participation rate in per cent for different mono-type consumer neighbourhoods. In addition to the average over 365 days and participants, the standard deviation between different participants is shown for each data point. For comparison, the results of a mixed neighbourhood (cf.~Figure~\ref{fig:savingVsNM}) are also presented.}
  	\label{fig:comparison_bill}
\end{figure}
\paragraph{Results:} The results in Section~\ref{sec:resultsComparison} and Section~\ref{sec:parRate_Error} are all based on a neighbourhood consisting of a mix of the three different consumer types (LOW, BASE, HIGH). Figure~\ref{fig:comparison_PARcons} shows the possible PAR reductions for mono-type neighbourhoods. To allow for comparison $M=25$ is kept constant. Furthermore, we use the same forecasting errors of $\epsilon_d=8\%$ and $\epsilon_w=10\%$ for the demand and renewable energy generation, respectively (cf.~Section~\ref{sec:parRate_Error}). All the simulations consider a scheduling period of a full year. 

We also calculated the average savings that are achieved by the participants of the DSM scheme. These results are presented in Figure~\ref{fig:comparison_bill} together with the reference of a mixed neighbourhood (cf.~Figure~\ref{fig:savingVsNM}) with forecasting errors.

\paragraph{Discussion:} 
\begin{figure}[thb]
	\centering
  	\includegraphics[width=0.69\textwidth]{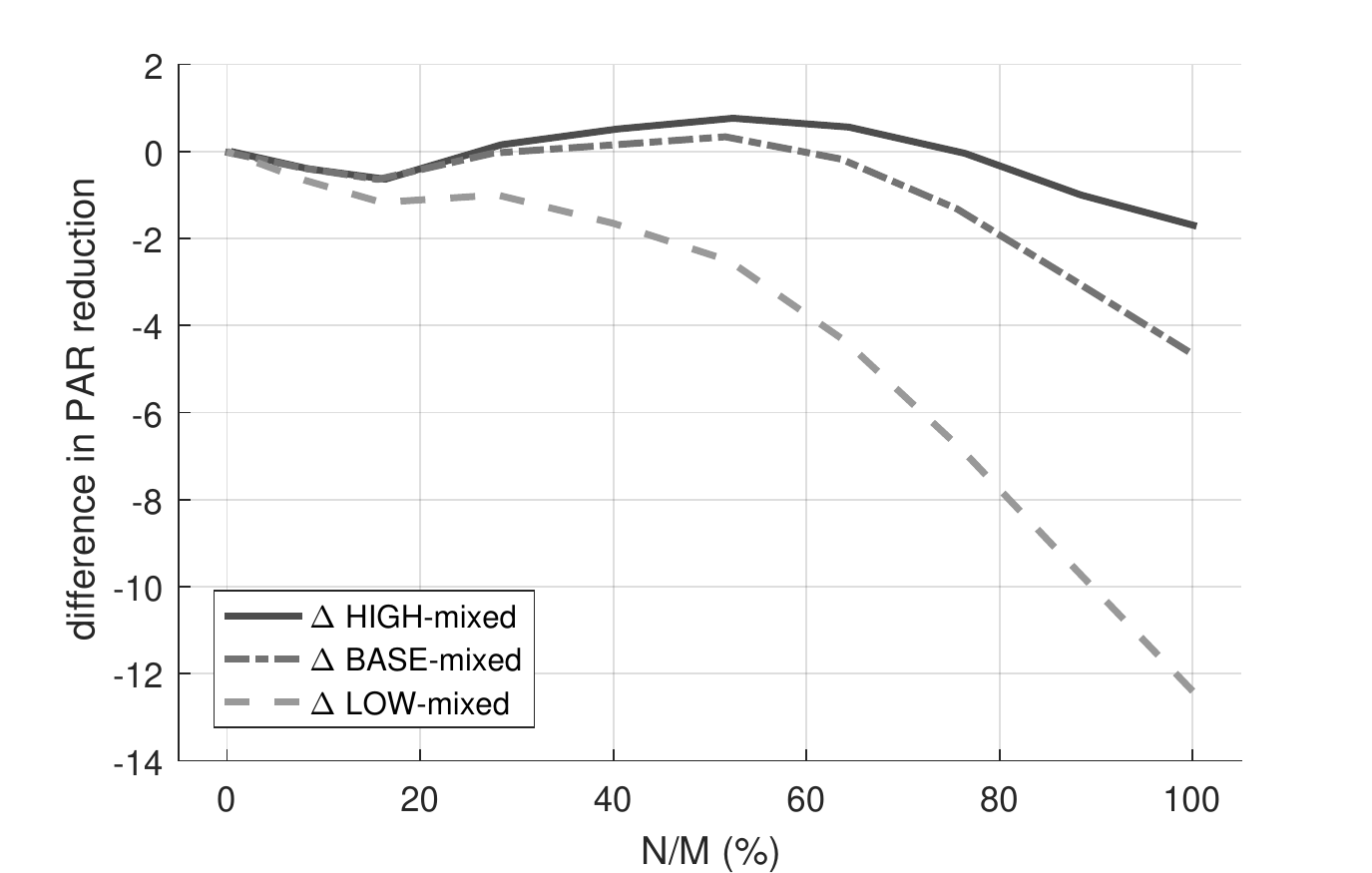}
  	\caption{\textit{Differences between neighbourhoods in peak-to-average ratio (PAR) reduction.} The values shown are based on the results represented in Figure~\ref{fig:comparison_PARcons}.}
  	\label{fig:discComparison_PARcons}
\end{figure}
When comparing different compositions of neighbourhoods, we can gain further insight into the conditions for which the DSM scheme works most efficiently. At first glance, Figure~\ref{fig:comparison_PARcons} reveals that given a low rate of participants in the scheme, the actual type of consumer is not crucial. Figure~\ref{fig:discComparison_PARcons} shows the difference between the respective results for mono-type neighbourhoods and a mixed neighbourhood. A closer look shows that mono-LOW communities are always worse in reducing the PAR value of the aggregated load than any of the other ones. The results in terms of both the mean PAR reduction and the standard deviation get even worse with more than two thirds of households participating in the scheme. Similar observations for mono-type neighbourhoods are found in~\cite{Soliman2014}.

For both mono-BASE and mono-HIGH neighbourhoods it can be observed that they perform better ($<1$\%) in an interval of medium participation rate than the mixed neighbourhood. Nevertheless, at $N=M$ the obtained PAR reduction is smaller by 1.8$\%$ and 4.5$\%$, respectively. Considering the variation of these mean PAR reduction values, it becomes clear that it is most beneficial to have a mixed-consumer neighbourhood.\\

Figure~\ref{fig:comparison_bill} shows the average bill reduction for the participants of the DSM scheme for different participation rates. Generally, they show the same behaviour already observed in Figure~\ref{fig:savingVsNM}. The influence of the proportionality factor in the billing scheme \eqref{eqn:billDSM} is clearly visible. Although the mixed-consumer neighbourhood achieves better PAR reduction, the average savings are almost identical to a mono-BASE neighbourhood. A neighbourhood that purely consists of HIGH consumers can save about 11$\%$ off the energy bill and is consistently most rewarding for the participants independent of the participation rate.


%% file: sec/5_conclusion.tex
In this paper, we propose\reviewDel{d} a demand-side management (DSM) scheme based on a discrete time dynamic game. Its purpose is to reduce the peak-to-average ratio (PAR) of the aggregated electricity load by scheduling the usage of individually owned (lithium-ion) energy storage systems. The utility company running the scheme, incentivises users to take part by offering fair financial benefits. To ensure realistic outcomes, an advanced battery model is employed. Furthermore, the integration of local energy generation in form of photovoltaic cells is taken into account. 

The DSM scheme is suitable for real-world implementation for four reasons: Firstly, it is based on a complete model of the neighbourhood including storage systems, local energy generation, and crucially forecasting errors of both demand and generation.
Secondly, computational costs to obtain schedules for the upcoming period are small and require only little amounts of memory. This was achieved by deriving a closed form solution for the best-response problem of an individual player. The ensuing iterative algorithm seems to converge exponentially towards \reviewDel{the }\reviewAdd{a }Nash-equilibrium and thus obtains the strategy profiles for one scheduling period in a fraction of a second.
Thirdly, the resulting schedules are robust with respect to the worst-case forecasting errors. Whereas the error weakens the effect of the PAR reduction by $\leq5.5\%$, the corresponding savings off the energy bill for the participants of the scheme are hardly changed.
Fourthly, we provide evidence that a neighbourhood that consists of various types of consumers performs best in such a DSM scheme. Since a mixed community is more probable than a mono-type community, this is a promising result.

A direct \reviewAdd{and in-depth }comparison to a DSM scheme with an underlying static game, revealed the advantages of the dynamic game approach. Players are overall more active and thus able to achieve distinctly better results. \reviewAdd{Further comparisons with the literature in terms of PAR reduction and computational costs show the superiority of our approach.}

\paragraph{\reviewAdd{Future Work:}}
In future work, we plan to corroborate our results with \reviewDel{a full probabilistic analysis of the forecasting errors. Moreover, we will extend the game--theoretic model to be able to deal with the finite-horizon effect. This will eventually lead to a scheduling mechanism which is insensitive to the starting time of the protocol.}\reviewAdd{an even more sophisticated approach to treat the uncertainties caused by the forecasts of demand and renewable energy generation. There are two main approaches that can be considered: (i) Robust (finite) game theory, which can be seen as a generalised version of a Bayesian game. (ii) A two step approach that first determines day-ahead schedules and then refines them throughout the day by using most recent data. The latter might be realised in a sliding window framework which would potentially eliminate the finite-horizon effects that can be encountered in our solutions. Furthermore, one could also think about modelling risks associated with the uncertainties directly in the utility function by means of the conditional value-at-risk measurement.

Whereas the current approach investigates a rather small community of households, it is worth to also explore the other end of the spectrum when the number of players becomes large. The method of choice here is mean field game theory in which the behaviour of the system is examined in the limit of an infinite number of players.}

%% file: sec/fin.tex
\begin{acknowledgements}
This work was supported by the Doctoral Training Alliance (DTA) Energy. The authors want to thank Jean-Christophe Nebel and Eckhard Pfluegel for helpful discussions. \reviewAdd{Furthermore, the authors want to thank the anonymous reviewers for their constructive feedback that helped to improve the manuscript.}
\end{acknowledgements}

\begin{spacing}{0.80}
\noindent{\small\rmfamily \textbf{Author Contributions}\ \;Matthias Pilz and Luluwah Al-Fagih conceived and designed the system; Matthias Pilz performed the theoretical analysis; Matthias Pilz implemented the software; Matthias Pilz performed the simulations; Matthias Pilz analysed the data; Matthias Pilz wrote the paper.}
\end{spacing}
\vspace{\baselineskip}
\begin{spacing}{0.80}
\noindent{\small\rmfamily \textbf{Conflicts of Interest}\ \;The authors declare no conflict of interest.}
\end{spacing}

%% file: sec/appendix.tex
\subsection{\reviewAdd{Proof of Theorem~\ref{thm:subgameNE}}}
\label{sec:thmProof}
We prove the theorem by contradiction. Suppose $\hat{a}^{t,T-1}$ is not a Nash equilibrium to the subgame $\left\{ G_1^{T-t},\dots,G_N^{T-t}\right\}$. Then, for some $n\in\mathcal{N}$, there must exist another strategy $\bar{a}_n^{t,T-1}$ with the corresponding sequence of states $\left\{\bar{s}_n^\tau\right\}_{\tau=t}^T$ such that 
	\begin{equation*}
		U_n^{T-t}\left(\left(\bar{s}_n^t,\hat{s}_{-n}^t\right), \left(\bar{a}_n^{t,T-1}, \hat{a}_{-n}^{t,T-1}\right)\right) \reviewDel{<}\reviewAdd{>} U_n^{T-t}\left(\hat{s}^t, \hat{a}^{t,T-1}\right)
	\end{equation*}
	Therefore, we obtain 
	\begin{align*}
	U_n\left(s_n^0, \left(\hat{a}^{0,t-1}, \left(\bar{a}_n^{t,T-1}, \hat{a}_{-n}^{t,T-1}\right)\right)\right) &= U_n^{T-t}\left(\left(\bar{s}_n^t,\hat{s}_{-n}^t\right), \left(\bar{a}_n^{t,T-1}, \hat{a}_{-n}^{t,T-1}\right)\right) \\ 
		&\hspace{0.4cm}\reviewDel{+}\reviewAdd{-} \sum_{\tau=0}^{t-1}g_n^\tau\left(\hat{s}_n^\tau, \left(\hat{a}_n^\tau, \hat{a}_{-n}^\tau\right)\right)\\
		&\reviewDel{<}\reviewAdd{>} U_n^{T-t}\left(\hat{s}^t, \hat{a}^{t,T-1}\right) \reviewDel{+}\reviewAdd{-} \sum_{\tau=0}^{t-1}g_n^\tau\left(\hat{s}_n^\tau, \left(\hat{a}_n^\tau, \hat{a}_{-n}^\tau\right)\right)\\
		&= U_n\left(\hat{s}^0,\hat{a}^{0,T-1}\right)
		= U_n\left(\hat{s}^0,\hat{a}\right)
	\end{align*}
	That is in contradiction to our assumption that $\hat{a}$ is a Nash equilibrium for the game $\left\{G_1, \dots, G_N\right\}$. Consequently, our assumption about $\hat{a}^{t,T-1}$ is proved to be false. Thus $\hat{a}^{t,T-1}$ indeed comprises a Nash equilibrium of the subgame $\left\{ G_1^{T-t},\dots,G_N^{T-t}\right\}$. \qed

\subsection{Battery Justification}
\label{sec:batteryJust}
Various companies produce home energy storage systems. To name just a few there are Mercedes, Tesla, BMW, Nissan and Powervault. Some of them are specialised in second life batteries taken from their electric cars, while others (such as Tesla) produce these batteries for their special purpose. As most manufacturers provide technical data sheets, we were able to run our simulations for the demand-side management scheme assuming that households are equipped with different batteries. 
\begin{figure}[bth]
	\centering
  	\includegraphics[width=0.79\textwidth]{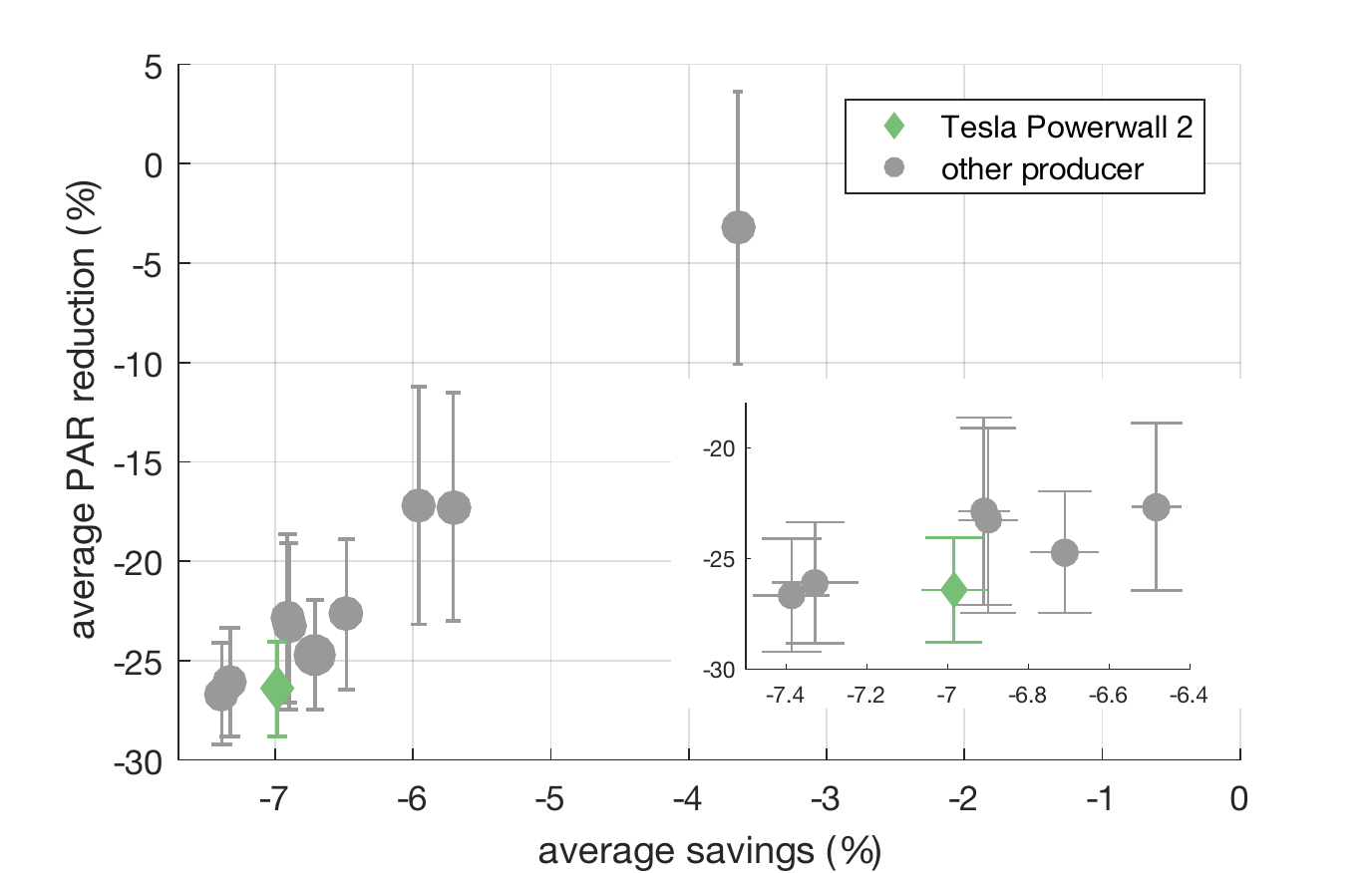}
  	\caption{\textit{Battery Justification.} The mean peak-to-average ratio reduction of the aggregated load over 365 days is plotted over the mean savings off the electricity bill for participants of the demand-side management scheme with different battery systems. For all runs we assumed $N/M=76\%$, $\epsilon_d=8\%$, $\epsilon_w=10\%$, and the same pricing parameter as introduced in Section~\ref{sec:results}. A close-up of the bottom left-hand corner is shown in a subplot. For all data points, we also provide the standard deviation in both variables.}
  	\label{fig:batteryJust}
\end{figure}
The results in Figure~\ref{fig:batteryJust} stem from scenarios with 76$\%$ participation rate, forecasting errors as used in Section~\ref{sec:parRate_Error}, and all participants with the exact same battery model. This is not supposed to compare different systems, but rather to show, that this battery (also employed in~\cite{Pilz2017}) can be taken as a representative of state-of-the-art technology. In this particular simulation run, it achieves a peak-to-average ratio reduction similar to the best in the field. Also the savings off the energy bill are close to the best competitors.